\documentclass[prodmode,acmjea]{acmsmall} 

\usepackage[ruled]{algorithm2e}

\usepackage{bbm}
\usepackage{amsfonts}
\usepackage{amsmath}
\usepackage{amssymb}
\usepackage{mathtools}
\usepackage{xspace}
\usepackage{color}
\usepackage{graphicx}
\usepackage{ifthen}
\usepackage{rotating}
\usepackage{url}
\usepackage{amsfonts}
\usepackage{adjustbox}
\usepackage{amsmath}
\usepackage{amssymb}

\newcommand{\wniosek}{corollary}

\def \testowncommand#1{#1}

\newcommand{\st}{\mathop{:}}

\newcommand{\concept}[1]{\textbf{#1{}}}
\newcommand{\bigO}[1]{\ensuremath{\mathit{\testowncommand O\!\left(#1\right)}}}
\newcommand{\bigTh}[1]{\ensuremath{\mathit{\testowncommand \Theta\!\left(#1\right)}}}
\newcommand{\bigOm}[1]{\ensuremath{\mathit{\testowncommand \Omega\!\left(#1\right)}}}

\newcommand{\vn}{\testowncommand \ensuremath{n}\xspace} 
\newcommand{\vm}{\testowncommand \ensuremath{{m}}\xspace} 

\newcommand{\vd}{\testowncommand \ensuremath{d}\xspace} 
\newcommand{\vP}{\testowncommand \ensuremath{{P}}\xspace} 
\newcommand{\vB}{\testowncommand \ensuremath{{B}}\xspace} 
\newcommand{\vM}{\testowncommand \ensuremath{{M}}\xspace} 
\newcommand{\vW}{\testowncommand \ensuremath{{W}}\xspace} 

\newcommand{\vT}{\testowncommand \ensuremath{{\tau}}\xspace} 
\newcommand{\rr}{\testowncommand \ensuremath{\mathrm{rr}}\xspace}
\newcommand{\beladys}{\testowncommand \ensuremath{\mathrm{MIN}}\xspace}
\newcommand{\lru}{\testowncommand \ensuremath{\mathrm{LRU}}\xspace}
\newcommand{\islru}{\testowncommand \ensuremath{\mathrm{ISLRU}}\xspace}
\newcommand{\rrmin}{\testowncommand \ensuremath{\mathrm{RRMIN}}\xspace}
\newcommand{\lazyis}{\testowncommand \ensuremath{\mathrm{LIS}}\xspace}
\newcommand{\ocs}{\testowncommand \ensuremath{\mathrm{ISMIN}}\xspace}
\newcommand{\ifnotempty}[2]{\ifthenelse{\equal{#1}{}}{}{#2}} 
\newcommand{\pro}[1]{\ensuremath{\mathit{\testowncommand \varphi_t\ifnotempty{#1}{\!\left(#1\right)}\xspace}}}
\newcommand{\prom}[2]{\ensuremath{\mathit{\testowncommand \varphi_t^{-#1}\ifnotempty{#2}{\!\left(#2\right)}}}}
\newcommand{\propl}[1]{\ensuremath{\mathit{\testowncommand \varphi_{t+1}\ifnotempty{#1}{\!\left(#1\right)}}}}
\newcommand{\TEM}{T_{\mathrm{EM}}}
\newcommand{\floor}[1]{\ensuremath{\left\lfloor#1\right\rfloor}}
\newcommand{\ceil}[1]{\ensuremath{\left\lceil#1\right\rceil}}

\renewcommand{\le}{\leqslant}
\renewcommand{\ge}{\geqslant}

\newcommand{\sift}{\mathit{sift}}
\newcommand{\ellrr}{\ell_{\rr}}
\newcommand{\ignore}[1]{}

\newcommand{\assign}{\ensuremath{\coloneqq}\xspace}

\newcommand{\seqset}[2]{\ensuremath{\left\{#1,\ldots,#2\right\}}\xspace}

\newcommand{\str}[2]{\ensuremath{#1\ldots#2}\xspace}

\newcounter{notedcaptioncounter}

\SetKwFor{ParFor}{for each vector processor $p$ in}{do in parallel}{endfor}
\SetKwFor{SIMDFor}{for each scalar processor $s$ in}{do in parallel}{endfor}
\SetKwInOut{Input}{Input}
\SetKwInOut{Assert}{Assert}
\SetKwInOut{Result}{Result}

\SetAlFnt{\small}
\SetAlCapFnt{\small}
\SetAlCapNameFnt{\small}
\SetAlCapHSkip{0pt}
\IncMargin{-\parindent}

\acmVolume{??}
\acmNumber{??}
\acmArticle{??}
\acmYear{??}
\acmMonth{4}

\begin{document}

\markboth{Tomasz Jurkiewicz and Kurt Mehlhorn}{On a Model of Virtual Address Translation}

\title{On a Model of Virtual Address Translation\thanks{A preliminary version of this paper appeared in ALENEX 2013.}}
\author{Tomasz Jurkiewicz
\affil{Max Planck Institute for Informatics, Saarbr\"ucken, Germany
\and
the Saarbr\"ucken Graduate School of Computer Science}
Kurt Mehlhorn
\affil{Max Planck Institute for Informatics, Saarbr\"ucken, Germany}
}

\begin{abstract}
Modern computers are not random access machines (RAMs).
They have a memory hierarchy, multiple cores, and a virtual memory.
We address the computational cost of the address translation in the virtual memory.

Starting point for our work on virtual memory is the observation that the analysis of some simple algorithms (random scan of an array, binary search, heapsort) in either the RAM model or the EM model (external memory model) does not correctly predict growth rates of actual running times.
We propose the VAT model (virtual address translation) to account for the cost of address translations and analyze the algorithms mentioned above and others in the model. 
The predictions agree with the measurements.
We also analyze the VAT-cost of cache-oblivious algorithms.

\end{abstract}


\terms{Design, Algorithms, Performance}


\acmformat{Tomasz Jurkiewicz and Kurt Mehlhorn, 2013. Computational Complexity of the Virtual Address Translation.}

\begin{bottomstuff}
A preliminary version of this article appeared in ALENEX 2013. The article is based on the first author's PhD-thesis.\\
Author's mail and web addresses:
\texttt{\{tojot, mehlhorn\}$\!\!\ \!\!\ \!\!\!$(at)mpi-inf.mpg.de},\\
\texttt{http://www.mpi-inf.mpg.de/\textasciitilde \{tojot, mehlhorn\}}
\end{bottomstuff}

\maketitle

\section{Introduction}

The role of models of computation in algorithmics is to provide abstractions of real machines for algorithm analysis.
Models should be mathematically pleasing and have a predictive value.
Both aspects are essential.
If the analysis has no predictive value, it is merely a mathematical exercise.
If a model is not clean and simple, researchers will not use it.
The standard models for algorithm analysis are the RAM (random access machine) model~\cite{Sheperson-Sturgis} and the EM (external memory) model~\cite{Aggarwal-Vitter}.

The RAM model is by far the most popular model.
It is an abstraction of the von Neumann architecture.
A computer consists of a control and processing unit and an unbounded memory.
Each memory cell can hold a word, and memory access as well as logical and arithmetic operations on words take constant time.
The word length is either an explicit parameter or assumed to be logarithmic in the size of the input.
The model is very simple and has a predictive value. 

Modern machines have virtual memory, multiple processor cores, an extensive memory hierarchy involving several levels of cache memory, main memory, and disks.
The external memory model was introduced because the RAM model does not account for the memory hierarchy, and hence, the RAM model has no predictive value for computations involving disks.
We give more details on both models in Section~\ref{sec:ram-and-em}.


This research started with a simple experiment.
We timed six simple programs for different input sizes, namely, permuting the elements of an array of size \vn, random scan of an array of size \vn, \vn random binary searches in an array of size \vn, heapsort of \vn elements,
introsort\footnote{Introsort is the version of quicksort used in modern versions of the STL. For the purpose of this paper, introsort is a synonym for quicksort.} of \vn elements, and sequential scan of an array of size \vn. 
For some of the programs, e.g., sequential scan through an array and quicksort, the measured running times agree very well with the predictions of the models.
\emph{However, the running time of random scan seems to grow as} $O(\vn\log\vn)$, \emph{and the running time of the binary searches seems to grow as} $O(\vn\log^2\vn)$, \emph{a blatant violation of what either model predicts.}
We give the details of the experiments in Section~\ref{experiments}. 

Why do measured and predicted running times differ?
Modern computers have virtual memories.
Each process has its own virtual address space $\{0,1,2,\ldots\}$.
Whenever a process accesses memory, the virtual address has to be translated into a physical address.
\emph{The translation of virtual addresses into physical addresses incurs cost}.
The translation process is usually implemented as a hardware-supported walk in a prefix tree, see Section~\ref{virtual memory} for details.
The tree is stored in the memory hierarchy, and hence, the translation process may incur cache faults.
The number of cache faults depends on the locality of memory accesses: the less local, the more cache faults.
The depth of the translation tree is logarithmic in the size of an algorithm's address space and hence, in the worst case, every memory access may lead to a logarithmic number of cache faults during the translation process.
For random scan and random binary searches, it apparently does. 

We propose an extension of the EM model, the VAT (Virtual Address Translation) model, that accounts for the cost of address translation, see Section~\ref{sec:VAT model}.
We show that we may assume that the translation process makes optimal use of the cache memory by relating the cost of optimal use with the cost under the LRU strategy, see Section~\ref{sec:VAT model}.
We analyze a number of programs, including the six mentioned above, in the VAT model and obtain good agreement with the measured running times, see Section~\ref{analysis of algorithms}.
We relate the cost of a cache-oblivious algorithm in the EM model to the cost in the VAT model, see Section~\ref{cache-oblivious algs}.
In particular,  cache-oblivious algorithms that do not need a tall-cache assumption incur no or little overhead.
In Section~\ref{discussion}, we address comments made by reviewers and readers of the paper. 
We close with some suggestions for further research and consequences for teaching, see Section~\ref{conclusions}.

\paragraph{Related Work:}

It is well-known in the architecture and systems community that virtual memory and address translation comes at a cost.
Many textbooks on computer organization, e.g.,~\cite{hennessy}, discuss virtual memories.
The papers by Drepper~\cite{Drepper2007,Drepper2008} describe computer memories, including virtual translation, in great detail.
\cite{AMD64} provides further implementation details.

The cost of address translation has received little attention from the algorithms community.
The survey paper by N.~Rahman~\cite{rahman} on algorithms for hardware caches and TLB summarizes the work on the subject.
She discusses a number of theoretical models for memory.
All models discussed in~\cite{rahman} treat address translation atomically, i.e., the translation from virtual to physical addresses is a single operation.
However, this is no longer true.
In 64-bit systems, the translation process is a tree walk.
Our paper is the first that proposes a theoretical model for address
translation and analyzes algorithms in this model.

\section{The Random Access Machine and the External Memory Machine}\label{sec:ram-and-em}

A RAM machine consists of a central processing unit and a memory.
The memory consists of cells indexed by nonnegative integers.
A cell can hold a bitstring.
The CPU has a finite number of registers, in particular an accumulator and an address register.
In any one step, a RAM can either perform an operation (simple arithmetic or boolean operations) on its registers or access memory.
In a memory access, the content of the memory cell indexed by the content of the address register is either loaded into the accumulator or written from the accumulator.
Two timing models are used: in the unit-cost RAM, each operation has cost one, and the length of the bitstrings that can be stored in memory cells and registers is bounded by the logarithm of the size of the input; in the logarithmic-cost RAM, the cost of an operation is equal to the sum of the lengths (in bits) of the operands, and the contents of memory cells and registers are unrestricted.

An EM machine is a RAM with two levels of memory.
The levels are referred to as cache and main memory or memory and disk, respectively.
We use the terms cache and main memory.
The CPU can only operate on data in the cache.
Cache and main memory are each divided into blocks of $B$ cells, and data is transported between cache and main memory in blocks.
The cache has size $M$ and hence consists of $M/B$ blocks; the main memory is infinite in size.
The analysis of algorithms in the EM-model bounds the number of CPU-steps and the number of block transfers.
The time required for a block transfer is equal to the time required by $\Theta(B)$ CPU-steps.
The hidden constant factor is fairly large, and therefore, the emphasis of the analysis is usually on the number of block transfers.

\section{Some Puzzling Experiments}\label{experiments}


We used the following seven programs in our experiments.
Let $A$ be an array of size \vn.
\begin{itemize}
\item permute: for $j \in [\vn-1 .. 0]$ do: $i \assign$ random$(0..j)$; swap$(A[i],A[j])$;
\item random scan: $\pi \assign$ random permutation; for $i$ from $0$ to $n-1$ do: $S \assign S+ A[\pi(i)]$;
\item \vn binary searches for random positions in $A$; $A$ is sorted for this experiment.
\item heapify
\item heapsort
\item quicksort
\item sequential scan
\end{itemize}

On a RAM, the first two, the last, and heapify are linear time \bigTh{\vn}, and the others are \bigTh{\vn\log\vn}.
Figure~\ref{data of first experiment} shows the measured running times for these programs divided by their RAM complexity; we refer to this quantity as \emph{normalized operation time}.
More details about our experimental methodology are available in Subsection \ref{sec:methodology}.
If RAM complexity is a good predictor, the normalized operation times should be approximately constant.
We observe that two of the linear time programs show linear behavior, namely, sequential access and heapify, that one of the \bigTh{\vn\log\vn} programs shows \bigTh{\vn\log\vn} behavior, namely, quicksort, and that for the other programs (heapsort, repeated binary search, permute, random access), the actual running time grows faster than what the RAM model predicts. 

\begin{center}
\emph{How much faster and why?}
\end{center}

Figure~\ref{data of first experiment} also answers the ``how much faster'' part of the question.
Normalized operation time seems to be a piecewise linear in the logarithm of the problem size; observe that we are using a logarithmic scale for the abscissa in this figure.
For heapsort and repeated binary search, normalized operation time is almost
perfectly piecewise linear, for permute and random scan, the piecewise linear
should be taken with a grain of salt.\footnote{We are leaving it as an open
  problem to give a satisfactory
  explanation for the bumpy shape of the graphs for permute and random access.}
The pieces correspond to the memory hierarchy.
\emph{The measurements suggest that the running times of permute and random scan grow like $\bigTh{n \log n}$ and the running times of heapsort and repeated binary search grow like $\bigTh{n \log^2 n}$.}

\begin{figure}[!h]
  	\begin{center}
	\begin{tabular}{cc}
	\adjustbox{valign=m}{\begin{sideways}\mbox{running time/RAM complexity}\end{sideways}}&
	\adjustbox{valign=m}{\includegraphics[width=0.9\textwidth]{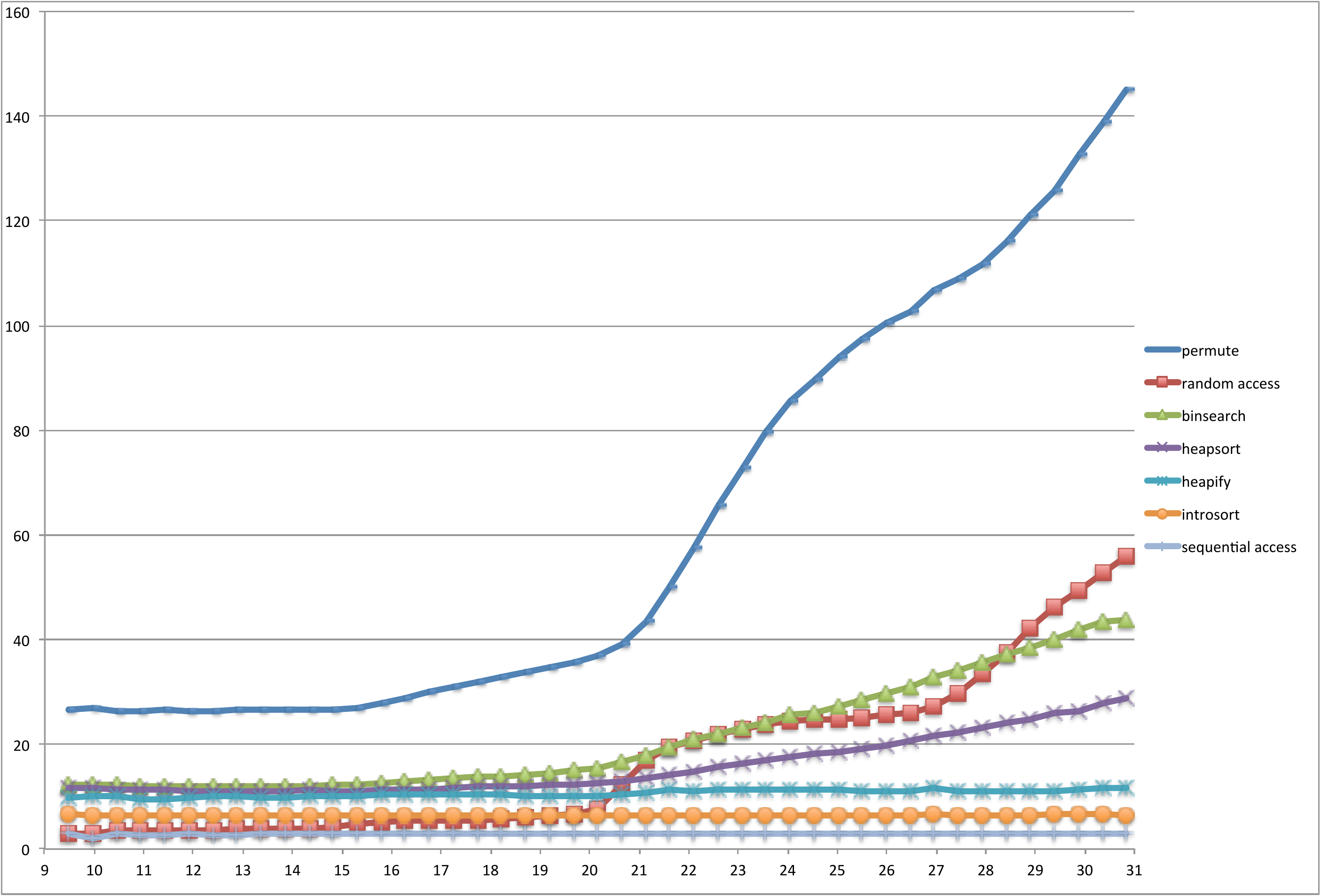}}\\
	&$\log(\text{input size})$
	\end{tabular}
  	\end{center}
\caption{
The abscissa shows the logarithm of the input size.
The ordinate shows the measured running time divided by the RAM-complexity (normalized operation time).
The normalized operation times of sequential access, quicksort, and heapify are
constant, the normalized operation times of the other programs are not.
}\label{data of first experiment}
\end{figure}

\subsection{Memory Hierarchy Does Not Explain It}

We argue in this section that the memory hierarchy does not explain the experimental findings. We give a detailed analysis of the cost of a random scan of an array of size $n$ in a hierarchical memory and relate it to the measured running time. We will see that the prediction by the model and the measured running times differ widely. A simpler argument for a one-level memory hierarchy will be given in Section~\ref{Relation to EM-model}. 

Let $s_i$, $i \ge 0$ be the size of the $i$-th level $C_i$ of the memory hierarchy; $s_{-1} = 0$.
We assume $C_i \subset C_{i+1}$ for all $i$.
Let $\ell$ be such that $s_\ell<\vn\le s_{\ell+1}$, i.e., the array fits into level $\ell+1$ but does not fit into level $\ell$.
For $i\le\ell$, a random address is in $C_i$ but not in $C_{i-1}$, with probability $(s_i - s_{i-1})/\vn$.
Let $c_i$ be the cost of accessing an address that is in $C_i$ but not in $C_{i-1}$.
The expected total cost in the external memory model is equal to
\begin{align*}
\TEM(\vn) \assign \vn\cdot\left(\frac{\vn-s_\ell}{\vn}c_{\ell+1}+\sum_{0\le i\le\ell} \frac{s_i-s_{i-1}}{\vn}c_i\right)= \vn c_{\ell+1} - \sum_{0 \le i \le \ell} s_i (c_{i+1} - c_i).
\end{align*}
This is a piecewise linear function whose slope is $c_{\ell+1}$ for $s_\ell < \vn\le s_{\ell+1}$.
The slopes are increasing but change only when a new level of the memory hierarchy is used.
Figure~\ref{random scan: measured by EM} shows the measured running time of random scan divided by EM-complexity as a function of the logarithm of the problem size.
Clearly, the figure does not show the graph of a constant
function.\footnote{The semi-log plot of a function of the form $n \mapsto (n \log (n/a))/(bn - c)$ with $a,b,c > 0$ is
  convex. Note that $\TEM(n) = bn - c$, where $b$ and $c$ depend on the level
  of the memory hierarchy required for $n$ and that Figure~\ref{data of first
    experiment} suggests that the actual running time grows like $n \log
  (n/a)$. The plot may be interpreted as the plot of a piecewise convex
  function and hence does not contradict the conclusion drawn from
  Figure~\ref{data of first experiment}.}

\newlength{\myskip}\settowidth{\myskip}{58}

\begin{figure}[h]
\begin{center}
	\begin{tabular}{cc}
	\adjustbox{valign=m}{\begin{sideways}\mbox{running time/EM complexity}\end{sideways}}&
	\adjustbox{valign=m}{\includegraphics[width=0.9\textwidth]{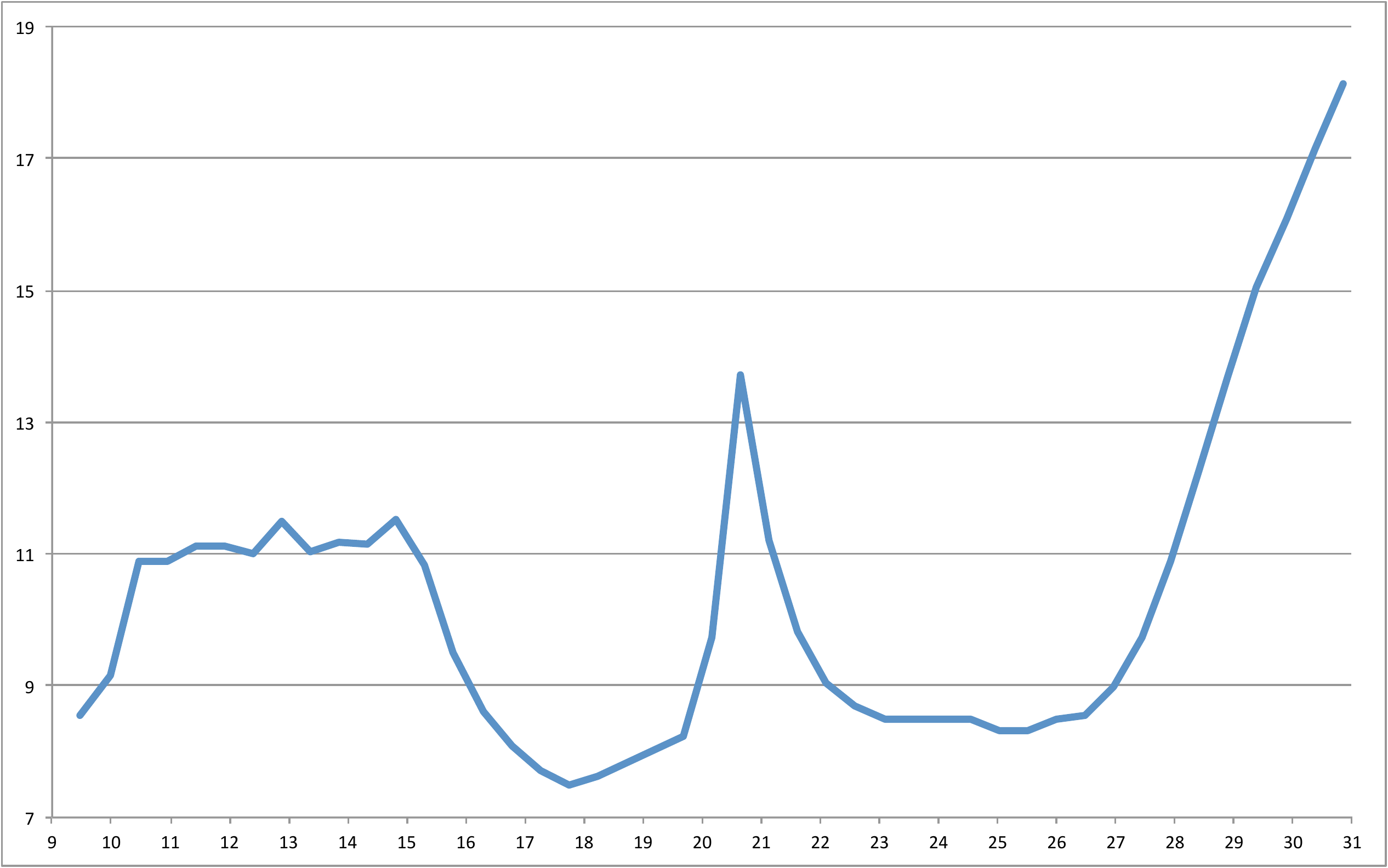}}\\
	&$\log(\text{input size})$
	\end{tabular}
\caption{\label{random scan: measured by EM} The running time of random scan divided by the EM-complexity. 
We used the following parameters for the memory hierarchy: the sizes are taken from the machine specification, and the access times were determined experimentally.}\medskip\footnotesize
\begin{tabular}{|c|c|c|c|}
\hline
Memory&Size&$\log(\text{maximum}$&Access Time\\
Level&&number of elements)&in Picoseconds\\
\hline
L1&32kiB&12\hspace{\myskip}\mbox{}&4080\\
L2&256kiB&15\hspace{\myskip}\mbox{}&4575\\
L3&12MiB&20,58&9937\\
RAM&&&38746\\
\hline
\end{tabular}
\end{center}
\end{figure}


\subsection{Methodology}\label{sec:methodology}
Programs used for the preparation of Figure \ref{data of first experiment} were compiled by gcc in version ``Debian 4.4.5-8'', and run on Debian Linux in version 6.0.3, on a machine with an Intel Xeon X5690 processor (3,46 GHz, 12MiB\footnote{KiB and MiB are modern, non-ambiguous notations for $2^{10}$ and $2^{20}$ bytes, respectively.
For more details, refer to \url{http://en.wikipedia.org/wiki/Binary_prefix}.} Smart Cache, 6,4 GT/s QPI). 
The caption of Figure~\ref{random scan: measured by EM} lists further machine parameters.
In each case, we performed multiple repetitions and took the minimum measurement for each considered size of the input data.
We chose the minimum because we are estimating the cost that must be incurred.
We also experimented with average or median; moreover, we performed the experiments on other machines and operating systems and obtained consistent results in each case.
We grew input sizes by factors of $1.4$ to exclude the influence of memory associativity, and we made sure that the largest problem size still fitted in the main memory to eliminate swapping.

For each experiment, we computed its normalized operation time, which we define as the measured execution time divided by the RAM complexity.
This way, we eliminate the known factors.
The resulting function represents cost of a single RAM-operation in relation to the problem size.

We use semi-log plots for showing normalized operation cost as a function of
the logarithm in the input size. In such a plot, linear functions of the
logarithm of the input size are easily identified as straight lines. 

\section{Virtual Memory}\label{virtual memory}

Virtual addressing was motivated by multi-processing.
When several processes are executed concurrently on the same machine, it is convenient and more secure to give each program a linear address space indexed by the nonnegative integers.
However, theses addresses are now virtual and no longer directly correspond to physical (real) addresses.
Rather, it is the task of the operating system to map the virtual addresses of all processes to a single physical memory.
The mapping process is hardware supported.


The memory is viewed as a collection of pages of $\vP = 2^p$ cells (= addressable units).
Both virtual and real addresses consist of an \emph{index} and an \emph{offset}.
The index selects a page and the offset selects a cell in a page.
The index is broken into $d$ segments of length $k = \log K$.
For example, for the long addressing mode of the processors of the AMD64 family (see http://en.wikipedia.org/wiki/X86-64)  the numbers are: $\vd = 4$, $k = 9$, and $p = 12$; the remaining 16 bits are used for other purposes. The choice of $k$ is not arbitrary. A page consists of $2^{12}$ bytes. An address consists of 8 bytes and hence a node of the translation tree requires $2^9 \cdot 2^3 = 2^{12}$ bytes. Thus nodes fit exactly into pages. 

\ignore{
\begin{figure}[!h]\centering
	\includegraphics[width=\textwidth]{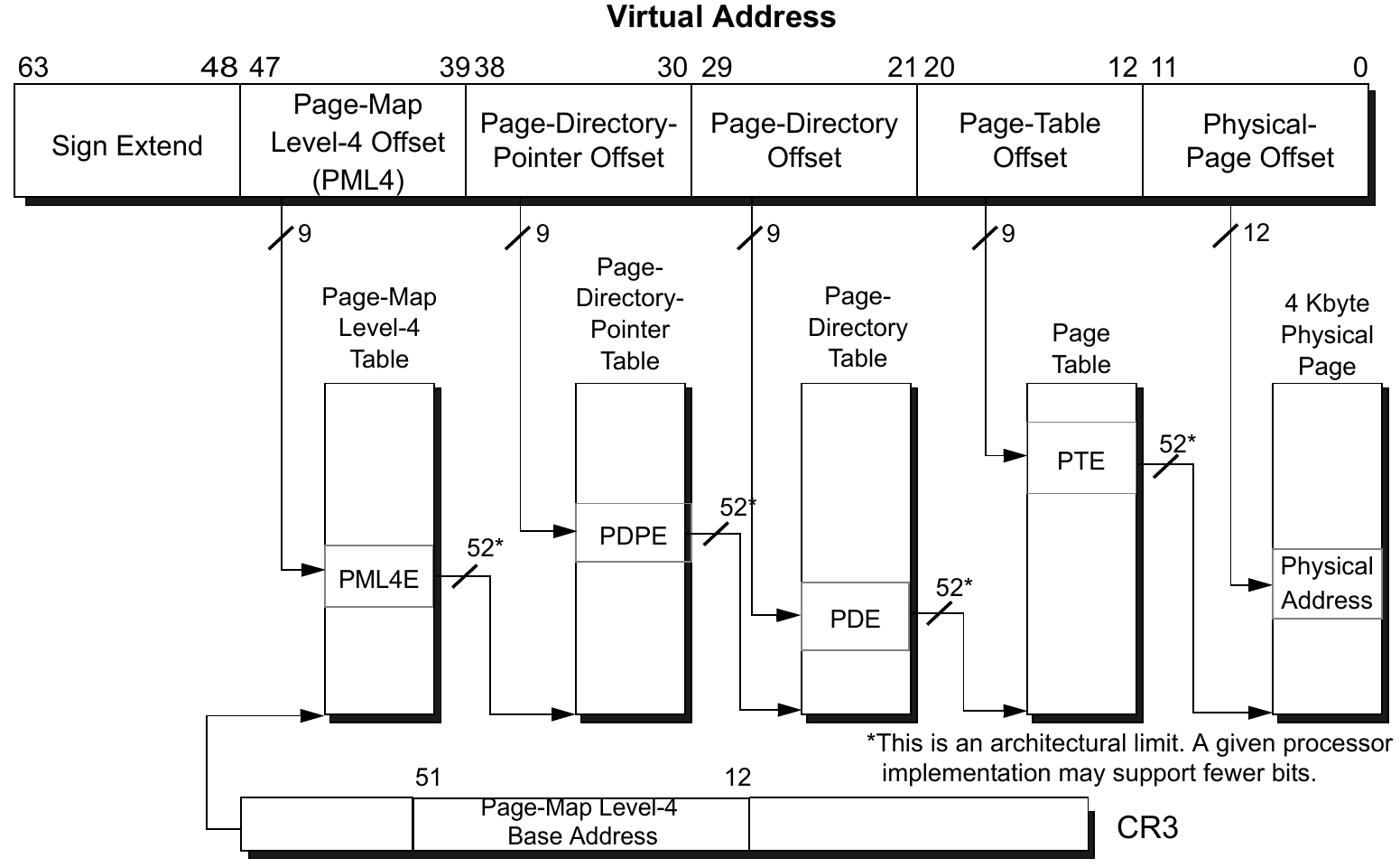}
\caption{Structure of a Virtual Address on a AMD64 class processor, figure from \cite{AMD64}. Note that offsets have distinct historical names, as this system was originally not designed to have multiple levels (Page-Map Level 4 Offset; Page-Directory Pointer Offset; Page-Directory Offset; and Page-Table Offset).
}\label{fig:translation}
\end{figure}
}
Logically, the translation process is a walk in a tree with outdegree $K$; this tree is usually called the page table~\cite{Drepper2008,hennessy}.
The walk starts at the root; the first segment of the index determines the child of the root, the second segment of the index determines the child of the child, and so on.
The leaves of the tree store indices of physical pages.
The offset then determines the cell in the physical address, i.e., offsets are not translated but taken verbatim.
Here quoting \cite{AMD64}:
\begin{quote}
``Virtual addresses are translated to physical addresses through hierarchical translation tables created and managed by system software.
Each table contains a set of entries that point to the next-lower table in the translation hierarchy.
A single table at one level of the hierarchy can have hundreds of entries, each of which points to a unique table at the next-lower hierarchical level.
Each lower-level table can in turn have hundreds of entries pointing to tables further down the hierarchy.
The lowest-level table in the hierarchy points to the translated physical page.

Figure \ref{fig:translation} on page \pageref{fig:translation} shows an overview of the page-translation hierarchy used in long mode.
Legacy mode paging uses a subset of this translation hierarchy.
As this figure shows, a virtual address is divided into fields, each of which is used as an offset into a translation table.
The complete translation chain is made up of all table entries referenced by the virtual-address fields.
The lowest-order virtual-address bits are used as the byte offset into the physical page.''
\end{quote}
\begin{figure}[!h]
	\includegraphics[width=\textwidth]{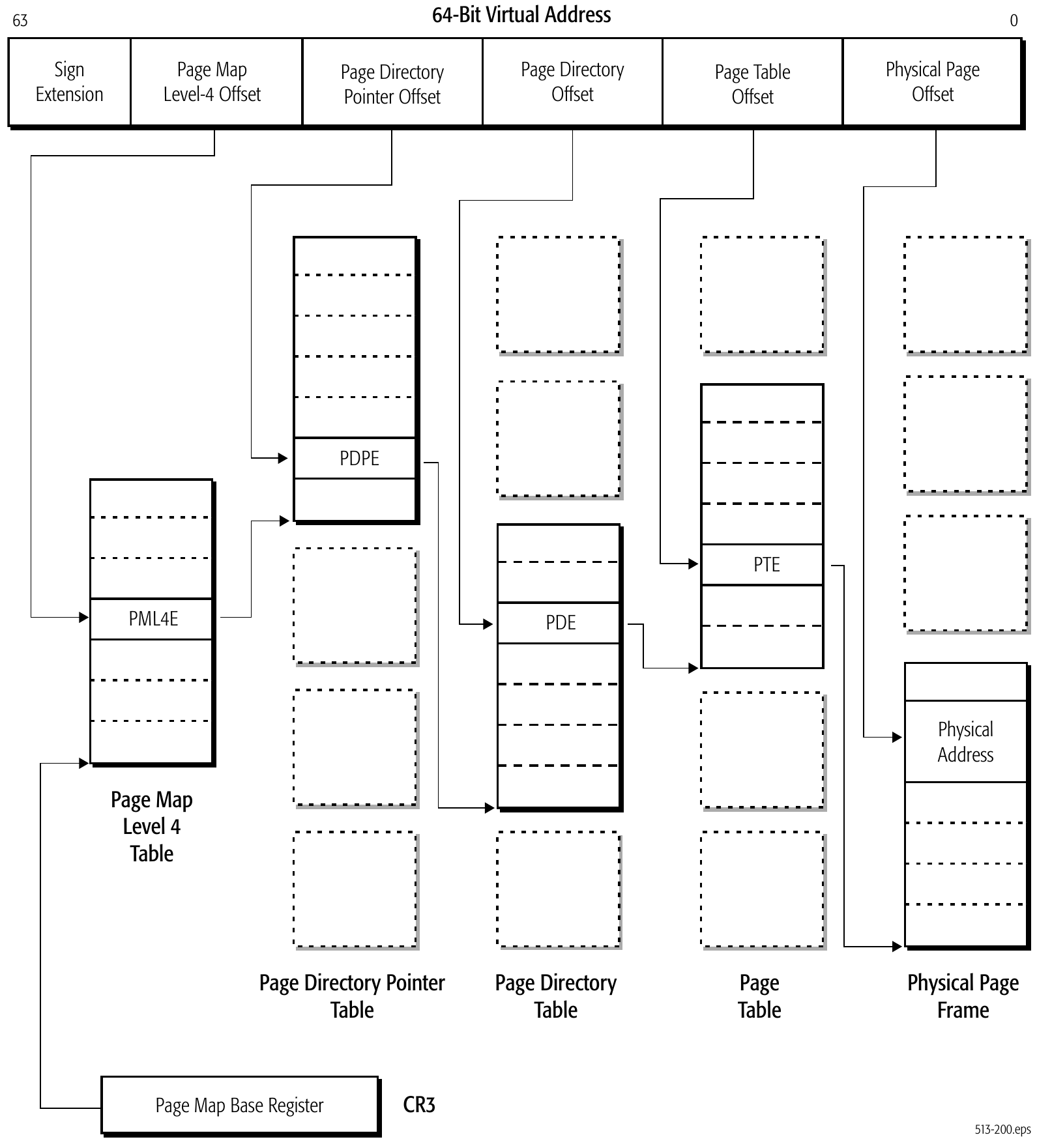}
\caption{Virtual to Physical Address Translation in AMD64, figure from \cite{AMD64}. Note that levels of the prefix tree have distinct historical names, as this system was originally not designed to have multiple levels (Page Map Level 4 Table; Page Directory Pointer Table; Page Directory Table; and Page Table).
}\label{fig:translation}
\end{figure}

Due to its size, the page table is stored in the RAM, but nodes accessed during the page table walk have to be brought to faster memory.
A small number of recent translations is stored in the translation-lookaside-buffer (TLB).
The TLB is a small associative memory that contains physical indices indexed by the virtual ones.
This is akin to the first level cache for data.
Quoting \cite{AMD64} further:
\begin{quote}\label{quo:locality-principle}
``Every memory access has its virtual address automatically translated into a physical address using the page-translation hierarchy.
\emph{Translation-lookaside buffers} (TLBs), also known as \emph{page-translation caches}, nearly eliminate the performance penalty associated with page translation.\linebreak
TLBs are special on-chip caches that hold the most-recently used virtual-to-physical address translations.
Each memory reference (instruction and data) is checked by the TLB. If the translation is present in the TLB, it is immediately provided to the processor, thus avoiding external memory references for accessing page tables.

TLBs take advantage of the \emph{principle of locality}. That is, if a memory address is referenced, it is likely that nearby memory addresses will be referenced in the near future.''
\end{quote}

\section{VAT, The Virtual Address Translation Model}\label{sec:VAT model}

\begin{figure}[!h]
\centering{\begin{tabular}{lr}
$P$       &  page size ($P = 2^p$)\\
$K$       & arity of translation tree ($K = 2^k$)\\
$d$        & depth of translation tree ($= \ceil{\log_K (\text{max used virtual address})}$)\\
$W$      & size of TC cache \\
$\tau$   & cost of a cache fault (number of RAM instructions)
\end{tabular}}
\caption{Notation \label{fig: essential quantities}}
\end{figure}

VAT machines are RAM machines that use virtual addresses.
We concentrate on the virtual memory of a single program.
Both real (physical) and virtual addresses are strings in $\{0, K-1\}^\vd\seqset{0}{\vP-1}$. Any such string corresponds to a number in the interval $[0,K^d P - 1]$ in a natural way.
The $\{0,K-1\}^\vd$ part of the address is called an \emph{index}, and its length \vd is an execution parameter fixed prior to the execution.
It is assumed that $\vd=\lceil\log_K(\text{maximum used virtual address}/\vP)\rceil$.
The $\seqset{0}{\vP-1}$ part of the address is called page offset and \vP is the page size.
The translation process is a tree walk.
We have a $K$-ary tree $T$ of height \vd.
The nodes of the tree are pairs $(\ell,i)$ with $\ell \ge 0$ and $i \ge 0$.
We refer to $\ell$ as the layer of the node and to $i$ as the number of the node.
The leaves of the tree are on layer zero and a node $(\ell, i)$ on layer $\ell \ge 1$ has $K$ children on layer $\ell-1$, namely the nodes $(\ell-1,Ki+a)$, for $a=0\ldots K-1$.
In particular, node $(\vd,0)$, the root, has children $(\vd-1,0),\ \ldots,\ (\vd-1,K-1)$.
The leaves of the tree are physical pages of the main memory of a RAM machine. In order to translate virtual address $x_{\vd-1}\ldots x_0 y$, we start in the root of $T$, and then follow the path described by \str{x_{\vd-1}}{x_0}. We refer to this path as the \emph{translation path} for the address. 
The path ends in the leaf $(0,\sum_{0 \leqslant i \leqslant \vd-1} x_i K^i)$. The offset $y$ selects the $y$-th cell in this page.

\begin{figure}[!h]
\begin{center}
\includegraphics[width=0.9\textwidth]{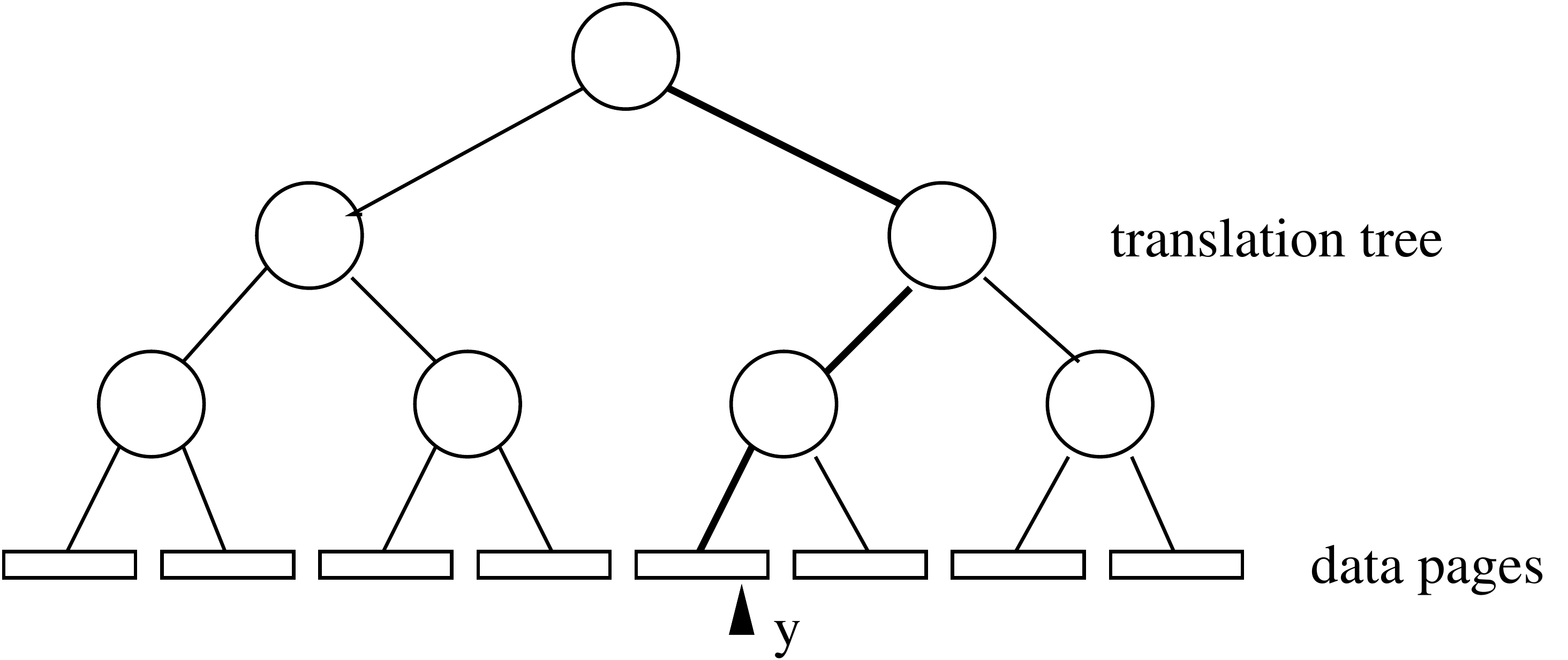}
\end{center}
\caption{\label{fig:translationtree} The pages holding the data are shown at the bottom and the translation tree is shown above the data pages. The translation tree has fan-out $K$ and depth $d$; here $K = 2$ and $d = 3$. The translation path for the virtual index 100 is shown. The offset $y$ selects a cell in the physical page with virtual index 100. The nodes of the translation tree and the data pages are stored in memory. Only nodes and data pages in fast memory (cache memory) can be accessed directly, nodes and data pages currently in slow memory have to brought into fast memory before they can accessed. Each such move is a cache fault. In the EM-model only cache faults for data pages are counted, in the VAT-model, we count cache faults for all nodes of the translation tree. }
\end{figure}

The translation process uses a translation cache TC that can store \vW nodes of the translation tree. Note that a node is either an internal node and then points to $K$ other nodes or is a leaf and then stores data. The cache is updated by insertions and evictions.
Let $a$ be a virtual address, and let $v_d,v_{d-1},\ldots,v_0$ be its translation path. Here, $v_d$ is the root of the translation tree, $v_d$ to $v_1$ are internal nodes of the translation tree, and $v_0$ is a data page.
Translating $a$ requires accessing all nodes of the translation path in order.
Only nodes in the TC can be accessed. 
The translation of $a$ ends when $v_0$ is accessed.

The \emph{number of cache faults incurred by the memory access} is the number of insertions performed during the translation process, and the \emph{cost of the memory access} is \vT times the number of cache faults.
The number of cache faults is at least the number of nodes of the translation path that are not
present in the cache at the beginning of the translation. 
Figure~\ref{fig: essential quantities} summarizes the notation. 

We close the introduction of the model with a trivial, but useful observation. 

\begin{lemma}\label{lem: large D} Let $D \ge 0$ and $i \ge 0$ be integers. The translation paths for addresses $i$ and $i + D$ differ in at least the last $\max(0,\ceil{\log_K (D/P)})$ nodes. \end{lemma}
\begin{proof} Let $u_d,\ldots,u_0$ and $v_d,\ldots,v_0$ be the translation paths for addresses $i$ and $i + D$, respectively. If $u_\ell = v_\ell$, then $D \le K^\ell P$. As a consequence, if $D > K^\ell P$, the translation path differ in the last $\ell + 1$ nodes.  \end{proof}

\subsection{Relation to EM-Model}\label{Relation to EM-model}

In the EM-model, one counts only cache faults caused by data pages, in the VAT-model one counts cache faults caused by all translation tree nodes.

A comparison of jumping scan and random scan illustrates the difference; this example is due to Jirka Kataijanen (personal communication). A jumping scan (with stride $P$) is a sequential scan of an array with stride $P$, i.e., cells $0$, $P$, $2P$, $3P$, \ldots are accessed in order. In the EM-model each access causes one page fault and hence the EM-cost of jumping scan and random scan are identical. 

In the VAT-model, a random scan is much more costly than a jumping scan. By Lemma~\ref{lem: random scan}, the cost of a random scan of an array of size $n$ is at least $\vT \vn\log_K(\vn/(\vP\vW))$. However, the cost of a jumping scan is at most $\tau n(1 + 1/(K-1))$. Observe (see Figure~\ref{fig:translationtree}) that $K$ subsequent data pages share the last internal of the translation path and hence the number of page faults caused by nodes of the translation tree is bounded by $\sum_{i \ge 1} n/K^i = n/K \cdot K/(K - 1) = n/(K-1)$. 

A second comparison is also informative. Assume that pages are accessed in random order and once a page is accessed all cells of the page are accessed. With respect to cache faults, $n$ such accesses are equivalent to $n/P$ random accesses. In the EM-model the number of cache faults is $n/P$, which is the same as for a linear scan, and in the VAT-model the number is at least $(\vn/P)\log_K(\vn/(\vP^2\vW))$.

\subsection{TC Replacement Strategies}
Since the TC is a special case of a cache in a classic EM machine, the following classic result applies.

\begin{lemma}[\cite{Sleator-Tarjan:List-Update-Paging,FrigoEtAl12}]\label{lem:lru<2min}
An optimal replacement strategy is at most by factor 2 better than LRU\footnote{LRU is a strategy that always evicts the Least Recently Used node.} on a cache of double size, assuming both caches start empty. 
\end{lemma}

This result is useful for upper bounds and lower bounds.
LRU is easy to implement.
In upper bound arguments, we may use any replacement strategy and then appeal to the Lemma.
In lower bound arguments, we may assume the use of LRU.
For TC caches, it is natural to assume the initial segment property.

\begin{definition}
An \concept{initial segment} of a rooted tree is an empty tree or a connected subgraph of the tree containing the root.
A TC has the \concept{initial segment property (ISP)} if the TC contains an initial segment of the translation tree.
A TC replacement strategy has ISP if under this strategy a TC has ISP at all times.
\end{definition}
\begin{proposition}
Strategies with ISP exist only for TCs with $\vW>\vd$.
\end{proposition}

ISP is important because, as we show later, ISP can be realized at no additional cost for LRU and at little additional cost for the optimal replacement strategy.
Therefore, strategies with ISP can significantly simplify proofs for upper and lower bounds.
Moreover, strategies with ISP are easier to implement.
Any implementation of a caching system requires some way to search the cache.
This requires an indexing mechanism.
RAM memory is indexed by the memory translation tree.
In the case of the TC itself, ISP allows to integrate the indexing structure into the cached content.
One only has to store the root of the tree at a fixed position or store the location of the root in a special register, called the Page Map Base Register in Figure~\ref{fig:translation}.

We establish the following relations in this section:
   \[ \text{MIN}(\vW)\leqslant\text{ISMIN}(\vW)\leqslant\text{ISLRU}(\vW)\leqslant\text{LRU}(\vW)\leqslant2\text{MIN}(\vW/2)\]
    \[ \text{MIN}(\vW)\geqslant\text{ISMIN}(\vW+d) \qquad
    \text{ISLRU}(\vW)\geqslant\text{LRU}(\vW+d); \]
here $W$ is the size of the translation cache, $d$ is the index length, LRU is the least-recently-used replacement strategy, MIN is the optimal cache replacement stragegy, ISLRU is LRU with the ISP-property, and ISMIN is the optimal replacement strategy with the ISP-property.

\subsection{Eager Strategies and the Initial Segment Property}
Before we prove an ISP analogue of Lemma~\ref{lem:lru<2min}, 
we need to better understand the behavior of replacement strategies with ISP. 
For classic caches, premature evictions and insertions do not improve efficiency.
We will show that the same holds true for TCs with ISP.
This will be useful because we will use early evictions and insertions in some of our arguments.

\begin{definition}\label{def:lazy}
A replacement strategy is \concept{lazy} if it performs an insertion of a missing node, only if the node is accessed directly after this insertion, and performs an eviction only before an insertion for which there would be no free cell otherwise. 
In the opposite case the strategy is \concept{eager}.
Unless stated otherwise, we assume that a strategy being discussed is lazy.
\end{definition}

Eager strategies can perform replacements before they are needed and can even insert nodes that are not needed at all.
Also, they can insert and re-evict, or evict and re-insert nodes during a single translation.
We eliminate this behavior \emph{translation by translation} as follows.
Consider a fixed translation and define the sets of \concept{effective evictions and insertions} as follows.
\begin{align*}
EE&=\{evict(a)\st \text{there are more $evict(a)$ than $insert(a)$ in the translation.}\}\\
EI&=\{insert(a)\st \text{there are more $insert(a)$ than $evict(a)$ in the translation.}\}
\end{align*}
Please note that in this case ``there are more'' means ``there is \emph{one} more'' as there cannot be two $evict(a)$ without an $insert(a)$ between them, or two $insert(a)$ without $evict(a)$.
\begin{proposition}\label{pro:effective-equivallence}
The effective evictions and insertions modify the content of the TC in the same way as the original evictions and insertions.
\end{proposition}

\begin{proposition}\label{pro:insertions-evictions}
During a single translation while a strategy with ISP is in use:
\begin{enumerate}
\item No node from the current translation path is effectively evicted, and all the nodes missing from the current translation path are effectively inserted.
\item If a node is effectively inserted, no ancestor or descendant of it is effectively deleted.
Subject to obeying the size restriction of the TC, we may therefore reorder effective insertions and effective deletions with respect to each other (but not changing the order of the insertions and not changing the order of the evictions). 
\end{enumerate}
\end{proposition}

\begin{lemma}\label{lem:eager-to-lazy}
Any eager replacement strategy with ISP can be transformed into a lazy replacement strategy with ISP with no efficiency loss. 
\end{lemma}
\begin{proof}
We modify the original evict/insert/access sequence \emph{translation by translation}.
Consider the current translation and let $EI$ and $EE$ be the set of effective insertions and evictions.
We insert the missing nodes from the current translation path exactly at the moment they are needed.
Whenever this implies an insertion into a full cache, we perform one of the lowest effective evictions, where lowest means that no children of the node are in the TC.
There must be such an effective eviction because, otherwise, the original sequence would overuse the cache as well.
When all nodes of the current translation path are accessed, we schedule all remaining effective evictions and insertions at the beginning of the next translation; first the evictions in descendant-first order and then the insertions in ancestor-first order.
The modified sequence is operationally equivalent to the original one, performs no more insertions, and does not exceed cache size.
Moreover, the current translation is now lazy.
\end{proof}

\subsection{ISLRU, or LRU with the Initial Segment Property}
Even without ISP, LRU has the property below.
\begin{lemma}\label{lem:lru-many-misses}
When the LRU policy is in use, the number of TC misses in a translation is equal to the layer number of the highest missing node on the translation path.
\end{lemma}
\begin{proof}
The content of the LRU cache is easy to describe.
Concatenate all translation paths and delete all occurrences of each node except the last.
The last $W$ nodes of the resulting sequence form the TC.
Observe that an occurrence of a node is only deleted if the node is part of a latter translation path.
This implies that the TC contains at most two incomplete translation paths, namely, the least recent path that still has nodes in the TC and the current path.
The former path is evicted top-down and the latter path is inserted top-down.
The claim now easily follows. 
Let $v$ be the highest missing node on the current translation path.
If no descendant of $v$ is contained in the TC, the claim is obvious.
Otherwise, the topmost descendant present in the TC is the first node on the part of the least recent path that is still in the TC.
Thus, as the current translation path is loaded into the TC, the least recent path is evicted top-down.
Consequently, the gap is never reduced. 
\end{proof}

The proof above also shows that whenever LRU detaches nodes from the initial segment, the detached nodes will never be used again.
This suggests a simple (implementable) way of introducing ISP to LRU.
If LRU evicts a node that still has descendants in the TC, it also evicts the descendants.
The descendants actually form a single path.
Next, we use Lemma~\ref{lem:eager-to-lazy} to make this algorithm lazy again.
It is easy to see that the resulting algorithm is the ISLRU, as defined next.
\begin{definition}
\concept{ISLRU} (Initial Segment preserving LRU) is the replacement strategy that always evicts the lowest descendant of the least recently used node.
\end{definition}
Due to the construction and Lemma \ref{lem:eager-to-lazy}, we have the following.
\begin{proposition}\label{pro:islru<lru}
ISLRU for TCs with $\vW>\vd$ is at least as good as LRU.
\end{proposition}
\begin{remark}
In fact, the proposition holds also for $\vW\leqslant\vd$, even though ISLRU no longer has ISP in this case.
\end{remark}

\subsection{\ocs: The Optimal Strategy with the Initial Segment Property}
\begin{definition}
\concept{ISMIN} (Initial Segment property preserving MIN) is the replacement strategy for TCs with ISP that always evicts from a TC the node that is not used for the longest time into the future among the nodes that are not on the current translation path and have no descendants. 
Nodes that will never be used again are evicted before the others in arbitrary descendant--first order.
\end{definition}
\begin{theorem}\label{the:isp-is-opt}
\ocs is an optimal replacement strategy among those with ISP.
\end{theorem}
\begin{proof} 
Let $R$ be any replacement strategy with ISP, and let $t$ be the first point in time when it departs from \ocs.
We will construct $R'$ with ISP that does not depart from \ocs, including time $t$, and has no more TC misses than $R$. 
Let $v$ be the node evicted by \ocs at time $t$.

We first assume that $R$ evicts $v$ at some later time $t'$ without accessing it in the interval $(t,t']$.
Then $R'$ simply evicts $v$ at time $t$ and shifts the other evictions in the interval $[t,t')$ to one later replacement. 
Postponing evictions to the next replacement does not cause additional insertions and does not break connectivity.
It may destroy laziness by moving an eviction of a node directly before its insertion.
In this case, $R'$ skips both.
Since no descendant of $v$ is in the TC at time $t$, and $v$ will not be used for the longest time into the future, none of its children will be added by $R$ before time $t'$; therefore, the change does not break the connectivity.

We come to the case that $R$ stores $v$ until it is accessed for the next time, say at time $t'$.
Let $a$ be the node evicted by $R$ at time $t$.
$R'$ evicts $v$ instead of $a$ and remembers $a$ as being \concept{special}.
We guarantee that the content of the TCs in the strategies $R$ and $R'$ differs only by $v$ and the current special node until time $t'$ and is identical afterwords. 
To reach this goal, $R'$ replicates the behavior of $R$ except for three situations.
\begin{enumerate}
\item If $R$ evicts the parent of the special node, $R'$ evicts the special node to preserve ISP and from then on remembers the parent as being special.
As long as only Rule 1 is applied, the special node is an ancestor of $a$.
\item If $R$ replaces some node $b$ with the current special node, $R'$ skips the replacement and from then on remembers $b$ as the special node.
Since $a$ will be accessed before $v$, Rule 2 is guaranteed to be applied, and hence, $R'$ is guaranteed to save at least one replacement.
\item At time $t'$, $R'$ replaces the special node with $v$, performing one extra replacement.
\end{enumerate}
We have shown how to turn an arbitrary replacement strategy with ISP into \ocs without efficiency loss.
This proves the optimality of \ocs.
\end{proof}
We can now state an ISP-aware extension of Lemma~\ref{lem:lru<2min}.

\begin{theorem}\label{the:equivalences}
\begin{align*}
\beladys(\vW)\leqslant \ocs(\vW) \leqslant \islru(\vW) \leqslant \lru(\vW) \leqslant 2 \beladys(\vW/2),
\end{align*}
where \beladys is an optimal replacement strategy and $A(s)$ denotes a number of insertions performed by replacement strategy $A$ to an initially empty TC of size $s>\vd$ for an arbitrary but fixed sequence of translations.
\end{theorem}
\begin{proof}
\beladys is an optimal replacement strategy, so it is better than \ocs.
\ocs is an optimal replacement strategy among those with ISP, so it is better than ISLRU.
ISLRU is better than LRU by Proposition~\ref{pro:islru<lru}.
$\lru(\vW)<2\beladys(\vW/2)$ holds by Lemma \ref{lem:lru<2min}.
\end{proof}

\subsection{Tighter Relationships}\ \\
Theorem~\ref{the:equivalences} implies $\lru(W) \le 2 \islru(W/2)$ and $\ocs(W) \le 2 \beladys(W/2)$. In this section, we sharpen both inequalities. 

\begin{lemma}
$\lru(\vW + \vd) \le \islru(\vW)$. 
\end{lemma}
\begin{proof}
\vd nodes are sufficient for LRU to store one extra path, hence, from the construction, LRU on a larger cache always stores a superset of nodes stored by ISLRU.
Therefore, it causes no more TC misses because it is lazy.
\end{proof}

\begin{theorem}\label{the:ismin<min}
$\ocs(\vW + \vd) \le \beladys(\vW)$.
\end{theorem}

The proof of Theorem~\ref{the:ismin<min} is lengthy. We will first derive some
properties of Belady's optimal algorithm $\beladys(\vW)$ and then transform any
$\beladys(\vW)$-strategy in two steps into a $\ocs(\vW + \vd)$-strategy. 

Recall that Belady's algorithm \beladys, also called the clairvoyant algorithm, is an optimal replacement policy.
The algorithm always replaces the node that will not be accessed for the longest time into the future.
An elegant optimality proof for this approach is provided in \cite{min-proof}.
\beladys does not differentiate between nodes that will not be used again.
Therefore, without loss of generality, let us from now on consider the descendant--first version of \beladys.
For any point in time, let us call all the nodes that are still to be accessed in the current translation \concept{the required nodes}.
The required nodes are exactly those nodes on the current translation path that are descendants of the last accessed node (or the whole path if the translation is only about to begin).
\begin{lemma}\label{lem:no-evictions}
\begin{enumerate}
\item\label{cla:no-evictions-2} Let $w$ be in the TC.
	As long as $w$ has a descendant $v$ in the TC that is not a required node, \beladys will not evict $w$.
\item\label{cla:no-evictions-3} If $\vW>\vd$, \beladys never evicts the root.
\item\label{cla:no-evictions-1} If $\vW>\vd$, \beladys never evicts a required node.
\end{enumerate}
\end{lemma}
\begin{proof}
Ad. \ref{cla:no-evictions-2}. If $v$ will be accessed ever again, then $w$ will be used earlier (in the same translation), and so, \beladys evicts $v$ before $w$.
If $v$ will never be accessed again, then \beladys evicts it before $w$ because it is the descendants--first version.
Ad. \ref{cla:no-evictions-3}. Either the TC stores the whole current translation path, and no eviction occurs, or there is a cell in the TC that contains a node off the current translation path; hence, the root is not evicted as it has a non-required descendant in the TC.
Ad. \ref{cla:no-evictions-1}. Either the TC stores the whole current translation path, or there is a cell $c$ in the TC with content that will not be used before any required node; hence, no required node is the node that will not be needed for the longest time into the future.
\end{proof}
\begin{\wniosek}\label{col:root-forever}
If $\vW>\vd$, \beladys inserts the root into the TC as the first insertion during the first translation, and never evicts it.
\end{\wniosek}
\begin{lemma}\label{lem:only-leaves}
If $\vW>\vd$, \beladys evicts only (non-required) nodes with no stored descendants or the node that was just used.
\end{lemma}
\begin{proof}
If \beladys evicts a node on the current translation path, it cannot be a descendant of the just translated node (Lemma \ref{lem:no-evictions}, claim \ref{cla:no-evictions-1}), it also cannot be an ancestor of the just translated node (Lemma \ref{lem:no-evictions}, claim \ref{cla:no-evictions-2}).
Hence, only the just translated node is admissible.
If the algorithm evicts a node off the current translation path, it must have no descendants (Lemma \ref{lem:no-evictions}, claim \ref{cla:no-evictions-2}).
\end{proof}
\begin{lemma}
If \beladys has evicted the node that was just accessed, it will continue to do so for all the following evictions in the current translation.
We will refer to this as \concept{round robin} approach.
\end{lemma}
\begin{proof}
If \beladys has evicted a node $w$ that was just accessed, it means that all the other nodes stored in the TC will be reused before the evicted node.
Moreover, all subsequent nodes traversed after $w$ in the current translation will be reused even later than $w$ if at all.
In case of $\vW>\vd$, the claim holds by Lemma \ref{lem:only-leaves}.
\end{proof}
\begin{\wniosek}\label{col:beladys-behaviour}
During a single translation, \beladys proceeds in the following way:\begin{enumerate}
\item It starts with the \concept{regular phase} when it inserts missing nodes of a connected path from the root up to some node $w$, as long as it can evict nodes that will not be reused before just used ones.
\item It switches to the \concept{round robin phase} for the remaining part of the path.
\end{enumerate}

It is easy to see that for $\vW>\vd$, in the path that was traversed in the round robin fashion, informally speaking, all gaps move up by one.
For each gap between stored nodes, the very TC cell that was used to store the node above the gap now stores the last node of the gap.
Storage of other nodes does not change.
This way, the number of nodes from this path stored in the TC does not change either.
However, it reduces numbers of stored nodes on side paths attached to the path.
\end{\wniosek}

In order to reach our goal, we will prove the following lemmas by modifying an optimal replacement strategy into intermediate strategies with no additional replacements.
\begin{lemma}\label{lem:ocp}
There is an eager replacement strategy on a TC of size $\vW+1$ that, except for a single special cell, has ISP and causes no more TC misses than an optimal replacement strategy on a TC of size \vW with no restrictions.
\end{lemma}
\begin{proof}
We introduce a replacement strategy \rrmin\footnote{Round Robin MIN}.
We add a special cell \rr to the TC, and we refer to the remaining \vW cells as \concept{regular TC}.
We will show that the cell \rr allows us to preserve ISP in the regular TC with no additional TC misses.
We start with an empty TC, and we run \beladys on a separate TC of size \vW on a side and observe its decisions.

We keep track of a partial bijection\footnote{A partial bijection on a set is a bijection between two subsets of the set.} \pro{} on the nodes of the translation tree.
We put one timestamp $t$ on every TC access and one more between every two accesses in the regular phase of \beladys.
We position evictions and insertions between the timestamps, at most one of each between two consecutive accesses.
At time $t$, \pro{} maps every node stored by \beladys in its TC to a node stored by \rrmin in its regular TC.
Function \pro{} always maps nodes to (not necessarily proper) ancestors in the memory translation tree.
We denote this as $\pro{a}\sqsubseteq a$, and in case of proper ancestors as $\pro{a}\sqsubset a$.
We say that $a$ is a witness for \pro{a}.
\begin{proposition}\label{pro:promotions}
Since the partial bijection \pro{} always maps nodes to ancestors, for every path of the translation tree, \rrmin always stores at least as many nodes as \beladys.
\end{proposition}
In order to prove the Lemma \ref{lem:ocp}, we need to show how to preserve the properties of the bijection \pro{} and ISP.
In accordance with Corollary \ref{col:beladys-behaviour}, \beladys inserts a number of highest missing nodes in the regular phase and uses the round robin approach on the remaining ones.

Let us first consider the case when \beladys has only the regular phase and inserts the complete path.
In this case, we substitute evictions and insertions of \beladys with these described below.

Let \beladys evict a node $a$.
If \pro{a} has no descendants, \rrmin evicts it.
In the other case, we find \pro{b} a descendant of \pro{a} with no descendants on his own.
\rrmin evicts \pro{b}, and we set $\propl{b}\assign\pro{a}$.
Clearly, we have preserved the properties of \propl{}\footnote{\propl{} is equal to \pro{} on all arguments not explicitly specified.}, and ISP holds.

Now let \beladys insert a new node $e$.
At this point, we know that both \rrmin and \beladys store all ancestors of $e$.
If \rrmin has not yet stored $e$, \rrmin inserts it, and we set $\propl{e}\assign e$.
If $e$ is already stored, it means it has a witness \prom{1}{e} that is a proper descendant of $e$.
We a find a sequence $e\sqsupset\prom{1}{e}\sqsupset\prom{2}{e}\sqsupset\ldots\sqsupset\prom{k}{e}=g$ that ends with $g$ \rrmin has not stored yet.
Such $g$ exists because \prom{1}{} is an injection on a finite set and is undefined for $e$.
We set $\propl{h}\assign h$ for all elements of the sequence except $g$.
\rrmin inserts the highest not stored ancestor $f$ of $g$, and we set $\propl{g}\assign f$.
Note that the inserted node $f$ might not be a required node.
Properties of \pro{} are preserved, and \rrmin did not disconnect the tree it stores.
Also, \rrmin performed the same number of evictions and insertions as \beladys.
Note as well that for all nodes on the translation path, \pro{} is identity.
Finally, Proposition \ref{pro:promotions} guarantees that all accesses are safe to perform at the time they were scheduled.

Now let us consider the case when \beladys has both regular and round robin phase.
Assume that the regular phase ends with the visit of node $v$.
At this point, \beladys stores the (nonempty for $\vW>\vd$ due to Corollary \ref{col:root-forever}) initial segment $p_v$ of the current path ending in $v$; it does not contain $v$'s child on the current path, and it contains some number (maybe zero) of required nodes.
Starting with $v$'s child, \beladys uses the round robin strategy.
Whenever it has to insert a required node, it evicts its parent.
Let $\ell_r$ and $\ellrr$ be the number of evictions in the regular and round robin phase, respectively.

\rrmin also proceeds in two phases. In the first phase, \rrmin simulates the regular phase as described above.
\rrmin also performs $\ell_r$ evictions in the first phase and \pro{} is the identity on $p_v$ at the end of the first phase;
this holds because $\pro{}$ maps nodes to ancestors and since \beladys contains $p_v$ in its entirety at the end of the regular phase.
Let $d'$ be the number of nodes on the current path below $v$;
\beladys stores $d' - \ellrr$ of them at the beginning of the round robin phase, which it does not have to insert, and it does not store $\ellrr$ of them, which it has to insert.
Since $\pro{}$ is the identity on $p_v$ after phase 1 of the simulation and maps the $d' -\ellrr$ required nodes stored by \beladys to ancestors, \rrmin stores at least the next $d'-\ellrr$ required nodes below $v$ in the beginning of phase 2 of the simulation.
In the round robin phase, \rrmin inserts the required nodes missing from the regular TC one after the other into $\rr$, disregarding what \beladys does.
Whenever \beladys replaces a node $a$ with its child $g$, in case of $\vW>\vd$ we fix $\pro{}$ by setting $\propl{g}\assign\pro{a}$.
By Proposition~\ref{pro:promotions}, \rrmin does no more evictions than \beladys.
Therefore, as it also preserves ISP in the regular TC, Lemma \ref{lem:ocp}
holds.\end{proof}

\begin{lemma}\label{lem:connected-TC}
There is a replacement strategy with ISP on a TC of size $\vW+\vd$, that causes no more TC misses than a general optimal replacement strategy on a TC of size \vW.
\end{lemma}
\begin{proof}
In order to prove the lemma, we will show how to use additional \vd regular cells in the TC to provide functionality of the special cell \rr while preserving ISP in the whole TC.
We run the \rrmin algorithm aside on a separate TC of size $\vW+1$, and
we introduce another replacement strategy, which we call \lazyis\footnote{Lazy strategy preserving the Initial Segments property}, on a TC of size $\vW+\vd$.
\lazyis starts with an empty TC where \vd cells are marked.
\lazyis preserves the following invariants.
\begin{enumerate}
\item The set of nodes stored in the unmarked cells by \lazyis is equal to the set of nodes stored in the regular TC by \rrmin.
\item The set of nodes stored in the marked cells by \lazyis contains the node stored in the cell \rr by \rrmin.
\item Exactly \vd cells are marked.
\item \lazyis has ISP.
\item No node is stored twice (once marked, once unmarked).
\end{enumerate}
Whenever \rrmin can replicate evictions/insertions of \lazyis without violating the invariants, it does.
Otherwise, we consider the following cases.
\begin{enumerate}
\item Let \rrmin in the regular phase evict a node $a$ that has marked descendants in \lazyis.
Then, \lazyis marks the cell containing $a$ and unmarks and evicts one of the marked nodes with no descendants that does not store the node stored in \rr.
Such a node exists because the only other case is that the marked cells contain all nodes of some path excluding the root, and the leaf is stored in \rr.
Therefore, $a$ is the root, but the root is never evicted due to ISP.
\item In the regular phase, \rrmin inserts a node $c$ to an empty cell while \lazyis already stores $c$ in a marked cell.
In this case, \lazyis unmarks the cell with $c$ and marks the empty cell.
\item In the round robin phase, \rrmin replaces the content of the cell \rr, \lazyis (if needed) replaces the content of an arbitrary marked node with no descendants that is not on the current translation path.
Since the root is always in the TC  and there are \vd marked cells, such a cell always exists.
ISP is preserved, as the parent of this node is already in the TC.
\end{enumerate}
At this stage, if we drop notions of \pro{} and marked nodes, \lazyis becomes an eager replacement strategy on a standard TC.
Therefore, we can use Lemma \ref{lem:eager-to-lazy} to make it lazy.
This concludes the proof of Lemma \ref{lem:connected-TC}.\end{proof}

Since \ocs is an optimal strategy with ISP, Theorem~\ref{the:ismin<min} follows from Lemma~\ref{lem:connected-TC}.


\begin{remark}
We believe that the requirement for \vd additional cache size is essentially optimal.
Consider the scenario when we access memory cells uniformly at random.
Informally speaking, \beladys will tend to permanently store the first $\log_K(\vW)$ levels of the translation tree because they are frequently used and will use a single cell to traverse the lower levels.
In order to preserve ISP, we need $\vd-\log_K(\vW)+1$ additional cells for
storing the current path. Thus only little improvement seems to be possible. 
\end{remark}

\begin{conjecture}
The strategy of storing higher nodes (Lemma \ref{lem:ocp}) and using extra \vd cells to not evict nodes from the current translation path (Lemma \ref{lem:connected-TC}) can be used to add ISP to any replacement strategy without efficiency loss.
\end{conjecture}

\section{Analysis of Algorithms}\label{analysis of algorithms}

In this section, we analyze the translation cost of some algorithms as a function of the problem size \vn and memory requirement $m$.
For all the algorithms analyzed, $\vm = \bigTh{\vn}$. 

In the RAM model, there is a crucial assumption that usually goes unspoken, namely, the size of a machine word is logarithmic in the number of memory cells used.
If the words were shorter, one could not address the memory.
If the words were longer, one could intelligently pack multiple values in one cell.
This technique can be used to solve NPC problems in polynomial time.
This effectively puts an upper bound on \vn, namely, $\vn<2^\text{word length}$, while asymptotic notations make sense only when \vn can grow to infinity.
However, this is not a limitation of  the RAM model, it merely shows that to handle bigger inputs, one needs more powerful machines.

In the VAT model there is also a set of assumptions on the model constants.
The assumptions bound \vn by machine parameters in the same sense as in the RAM model.
However, unlike in the RAM model, they can hardly go unspoken.
We call them the \emph{asymptotic order relations} between parameters.
The assumptions we found necessary for the analysis to be meaningful are as follows:
\begin{enumerate}
\item\label{ass:tdp} $1 \leqslant \vT d \leqslant P$; moving a single
  translation path to the TC costs more than a single instruction, but not more
  than size-of-a-page many instructions, i.e., if at least one instruction is performed for each cell in a page, the cost of translating the index of the page can be amortized. 
\item $K \ge 2$, i.e., the fanout of the translation tree is at least two.
\item $\vm/\vP \le K^\vd \le 2\vm/\vP$, i.e., the translation tree suffices to translate all addresses but is not much larger.
	As a consequence, $\log (\vm/\vP) \le \vd \log K = \vd k \le \log (2m/P) = 1 + \log(\vm/\vP)$, and hence, $\log_K (\vm/\vP) \le \vd \le \log_K (2m/P) = (1 + \log (\vm/\vP))/k$.  
\item $\vd \le \vW < \vm^\theta, \text{ for }\theta\in(0,1)$, i.e., the translation cache can hold at least one translation path, but is significantly smaller than the main memory.
\end{enumerate}

\subsection{Sequential Access} We scan an array of size \vn, i.e., we need to translate addresses $b$, $b+1$, \ldots, $b + \vn - 1$ in this order, where $b$ is the base address of the array.
The translation path stays constant for \vP consecutive accesses, and hence, at most $2 \vn/\vP$ indices must be translated for a total cost of at most $\vT\vd(2 + \vn/\vP)$.
By assumption \eqref{ass:tdp}, this is at most $\vT\vd(\vn/\vP + 2)\le \vn + 2\vP$.

The analysis can be sharpened significantly.
We keep the current translation path in the cache, and hence, the first translation incurs at most $\vd$ faults. 
The translation path changes after every $\vP$-th access and hence changes at most a total of $\ceil{\vn/\vP}$ times.
Of course, whenever the path changes, the last node changes.
The next to last node changes after every $K$-th change of the last node 
 and hence changes at most $\ceil{\vn/(\vP K)}$ times.
In total, we incur $$\vd + \sum_{0 \le i \le \vd} \ceil{\frac{\vn}{\vP K^i}} < 2\vd + \frac{K}{K-1} \frac{\vn}{\vP}$$ cache faults; of these faults, at most $1 + n/P$ are caused by data pages and the remaining ones are causes by internal nodes of the translation tree.
The cost is therefore bounded by $2\tau d + 2\tau n/P \le 2\vP + 2\vn/\vd$, which is asymptotically
smaller than the RAM complexity.

\subsection{Random Access}
In the worst case, no node of any translation path is in the cache.
Thus the total translation cost is bounded by $\vT \vd \vn$.
This is at most $\tau n \log_K (2\vn/\vP))$.

We will next argue a lower bound.
We may assume that the TC satisfies the initial segment property.
The translation path ends in a random leaf of the translation tree.
For every leaf, some initial segment of the path ending in this leaf is cached.
Let $u$ be an uncached node of the translation tree of minimal depth, and let $v$ be a cached node of maximal depth.
If the depth of $v$ is larger by two or more than the depth of $u$, then it is better to cache $u$ instead of $v$ (because more leaves use $u$ instead of $v$).
Thus, up to one, the same number of nodes is cached on every translation path, and hence, the expected length of the path cached is at most $\log_K\vW$, and hence, the expected number of faults during a translation is $\vd - \log_K \vW$.
The total expected cost is therefore at least
$\vT\vn(\vd-\log_K\vW)\geqslant\vT\vn\log_K\vn/(\vP\vW)$, 
which is asymptotically larger than the RAM complexity.

\begin{lemma}\label{lem: random scan}
The memory access cost of a random scan of an array of size $n$ is at least $\vT \vn\log_K (\vn/(\vP\vW))$ and at most $\vT n \log_K (2\vn/\vP)$. 
\end{lemma}

\subsection{Binary Search} We do \vn binary searches in an array of length \vn.
Each search  searches for a random element of the array.
For simplicity, we assume that \vn is a power of two minus one.
A binary search in an array is equivalent to a search in a balanced tree where the root is stored in location $n/2$, the children of the root are stored in locations $n/4$ and $3n/4$, and so on.
We cache the translation paths of the top $\ell$ layers of the search tree and the translation path of the current node of the search.
The top $\ell$ layers contain $2^{\ell + 1} - 1$ vertices, and hence, we need to store at most $\vd 2^{\ell + 1}$ nodes\footnote{We use vertex for the nodes of the search tree and node for the nodes of the translation tree.} of the translation tree.
This is feasible if $d 2^{\ell + 1} \le \vW$. For next two paragraphs, let $\ell = \log (\vW/2\vd)$.

Any of the remaining $\log n - \ell$ steps of the binary search causes at most $d$ cache faults.
Therefore, the total cost per search is bounded by 
\begin{align*}
\vT \vd (\log \vn - \ell) \le& \vT\log_K(2\vn/\vP)(\log n - \ell) = \vT \log_K\frac{2\vn}{\vP} \log \frac{2\vn \vd}{\vW}.
\end{align*}

This analysis may seem unrefined.
After all once the search leaves the top $\ell$ layers of the search tree, addresses of subsequent nodes differ only by $n/2^{\ell}$, $n/2^{\ell + 1}$, \ldots, 1.
However, we will next argue that the bound above is essentially sharp for our caching strategy. In a second step, we will extend the bound to all caching strategies.
By Lemma~\ref{lem: large D}, if two virtual addresses differ by $D$, their translation paths differ in the last $\ceil{\log_K (D/\vP)}$ nodes.
Thus, the scheme above incurs at least 
\begin{align*}
\sum_{\ell \le i \le \log n - p}& \ceil{ \frac{1}{k} \log \frac{n}{2^i P}} \ge \sum_{0 \le j \le \log n - \ell - p} \frac{1}{k} \log 2^j \ge\\
\ge&\frac{1}{2k} (\log n - \ell - p)^2 = \frac{1}{2} \log_K \frac{2nd}{\vP\vW} \log \frac{2nd}{\vP\vW}
\end{align*}
cache faults.
We next show that it essentially holds true for any caching strategy. 

By Theorem~\ref{the:equivalences}, we may assume that ISLRU is used as the cache replacement strategy, i.e., TC contains top nodes of the recent translation paths.
Let $\ell = \ceil{\log (2\vW)}$.
There are $2^\ell \ge 2\vW$ vertices of depth $\ell$ in a binary search tree.
Their addresses differ by at least $\vn/2^\ell$, and hence, for any two such addresses, their translation paths differ in at least the last $z = \ceil{\log_K(n/(2^\ell \vP)}$ nodes.
Call a node at depth $\ell$ \emph{expensive} if none of the last $z$ nodes of its translation path are contained in the TC and \emph{inexpensive} otherwise.
There can be at most \vW inexpensive nodes, and hence, with probability at least 1/2, a random binary search goes through an expensive node, call it $v$, at depth $\ell$.
Since ISLRU is the cache replacement strategy, the last $z$ nodes of the translation path are missing for all descendants of $v$.
Thus, by the argument in the preceding paragraph, the expected number of cache misses per search is at least 
\begin{align*}
\frac{1}{2} \sum_{\ell \le i \le \log n - p}& \ceil{ \frac{1}{k} \log \frac{n}{2^i P}} \ge \frac{1}{4} \log_K \frac{2nd}{\vP\vW} \log \frac{2nd}{\vP\vW}
\end{align*}

\begin{lemma} The memory access cost of $n$ random binary searches in an array of size $n$ is at most $\vT \log_K\frac{2\vn}{\vP} \log \frac{2\vn \vd}{\vW}$ and at least $\frac{\vT}{4} \log_K \frac{2nd}{\vP\vW} \log \frac{2nd}{\vP\vW}$
 \end{lemma}

We know from cache-oblivious algorithms that the van-Emde Boas layout of a search tree improves locality.
We will show in Section~\ref{cache-oblivious algs} that this improves the translation cost.

\subsection{Heapify and Heapsort}
We prove a bound on the translation cost of heapify. The following proposition generalizes the analysis of sequential scan. 

\begin{definition}
\concept{Extremal translation paths} of \vn consecutive addresses are the paths to the first and last address in the range.
\concept{Non-extremal nodes} are the nodes on translation paths to addresses in the range that are not on the extremal paths.
\end{definition}

\begin{proposition}\label{pro:scan-cost}
A sequence of memory accesses that gains access to each page in a range causes at least one TC miss for each non-extremal node of the range.
If the sequence of pages in the range \vn is accessed in the decreasing order, this bound is matched by storing the extremal paths and dedicating $\log_K(\vn/\vP)$ cells in the TC for the required translations.
\end{proposition}

\begin{proposition}\label{pro:count} Let $n$, $\ell$, and $x$ be nonnegative integers. 
The number of non-extremal nodes in the union of the translation paths of any $x$ out of \vn consecutive addresses is at most  $$x \ell + \frac{2\vn}{\vP K^\ell}.$$
Moreover, there is a set of $x = \ceil{n/(P K^\ell)}$ addresses such that the union of the paths has size at least $x(\ell + 1) + d - \ell$. 
\end{proposition}
\begin{proof}
The union of the translation paths to all $n$ addresses contains at most $\vn/\vP$ non-extremal nodes on the leaf level (= level 0) of the translation tree.
On level $i$, $i \ge 0$, from the bottom, it contains at most $\vn/(\vP K^i)$ non-extremal nodes.

We overestimate the size of the union of $x$ translation paths by counting one node each on levels $0$ to $\ell - 1$ for every translation path and all non-extremal nodes contained in all the $n$ translation paths on the levels above. 
Thus, the size of the union is bounded by $$x \ell + \sum_{\ell \le i \le d} \vn/(\vP K^i) < x \ell + \frac{K}{K-1} \frac{n}{P K^\ell}\leqslant x \ell + \frac{2\vn}{\vP K^\ell}.$$

A node on level $\ell$ lies on the translation path of $K^\ell P$ consecutive addresses.
Consider addresses $z + i P K^\ell$ for $i = 0,1,\ldots, \ceil{n/PK^\ell} - 1$, where $z$ is the smallest in our set of $n$ addresses.
The translation paths to these addresses are disjoint from level $\ell$ down to level zero and use at least one node on levels $\ell+1$ to $d$.
Thus, the size of the union is at least $x(\ell+1) + d - \ell$. 
\end{proof}

An array $A[1..n]$ storing elements from an ordered set is heap-ordered if $A[i] \le A[2i]$ and $A[i] \le A[2i+1]$ for all $i$ with $1 \le i \le \floor{n/2}$.
An array can be turned into a heap by calling operation $\sift(i)$ for $i = \floor{n/2}$ down to $1$.
$\sift(i)$ repeatedly interchanges $z = A[i]$ with the smaller of its two children until the heap property is restored.
We use the following translation replacement strategy.
Let $z =\min(\log n,\floor{(\vW - 2\vd-1)/\floor{\log_K(\vn/\vP)}}-1)$.
We store the extremal translation paths ($2\vd-1$ nodes), non-extremal parts of the translation paths for $z$ addresses $a_0$, \ldots, $a_{z-1}$, and one additional translation path $a_\infty$ ($\floor{\log_K(\vn/\vP)}$ nodes for each).
The additional translation path is only needed when $z \neq \log n$.
During the siftdown of $A[i]$, $a_0$ is equal to the address of $A[i]$, $a_1$ is the address of one of the children of $i$ (the one to which $A[i]$ is moved, if it is moved), $a_2$ is the address of one of the grandchildren of $i$ (the one to which $A[i]$ is moved, if is moved two levels down), and so on.
The additional translation path $a_\infty$ is used for all addresses that are more than $z$ levels below the level containing $i$. 

Let us upper bound the number of the TC misses.
Preparing the extremal paths causes up to $2d+1$ misses.
Next, consider the translation cost for $a_i$, $0 \le i \le z-1$.
$a_i$ assumes $\vn/2^i$ distinct values.
Assuming that siblings in the heap always lie in the same page\footnote{This assumption can be easily lifted by allowing an additional constant in the running time or in the TC size.},
the index (= the part of the address that is being translated) of each $a_i$ decreasses over time, and hence, Proposition \ref{pro:scan-cost} bounds the number of TC misses to the number of the non-extremal nodes in the range.
We use Proposition~\ref{pro:count} to count them.
For $i \in \{0,\ldots,p\}$, we use the Proposition with $x = \vn$ and $\ell = 0$ and obtain a bound of  $$\frac{2\vn}{\vP} = \bigO{\frac{\vn}{\vP}}$$ TC misses.
For $i$ with $p+(\ell-1) k<i\leqslant p+\ell k$, where $\ell \geqslant 1$ and $i\leqslant z - 1$, we use the proposition with $x = \vn/2^i$ and obtain a bound of at most
\begin{align*}
\frac{\vn}{2^i}\cdot\ell + \frac{2\vn}{\vP K^\ell} = \bigO{\frac{\vn}{2^i}\cdot\ell + \frac{2\vn}{2^i}} = \bigO{\frac{\vn}{2^i}(\ell + 2)} = \bigO{\vn \frac{i}{2^i}}
\end{align*}
TC misses.
There are $\vn/2^z$ siftdowns starting in layers $z$ and above; they use $a_\infty$.
For each such siftdown, we need to translate at most $\log n$ addresses, and each translation causes less than \vd misses.
The total is less than $\vn (\log \vn) \vd/2^a$.
Summation yields
\begin{align*}
2\vd+1 + (p+1) \bigO{\frac{\vn}{\vP}} +\!\! \sum_{p < i \le z-1}\!\!\!\! \bigO{\vn \frac{i}{2^i}} + \frac{ \vn d \log \vn}{2^z} =\bigO{\vd + \frac{\vn p}{\vP} + \frac{ \vn d \log \vn}{2^z}}.
\end{align*}
For any realistic values of the parameters, the third term is insignificant,
hence, the cost is \bigO{\vT(d+\frac{\vn p}{\vP})}.  

We next prove the corresponding lower bound under the additional assumption that $\vW < \frac12\vn/\vP$.
At least one address must be completely translated; hence the cost of \bigOm{\vT\vd}.
The addresses in $a_0 \ldots a_{p-1}$ assume at least one address per page in the subarray $[n/2 .. n]$ because $a_i$ can never jump by more than $2^{i+1}$.
First, the addresses are swept by $a_0$, then by $a_1$, and so on, and no other accesses to the subarray occur in the meantime. 
Hence, if the LRU strategy is in use and $\vW < \frac12\vn/\vP$, there are at least $p \vn/(2\vP)$ TC misses to the lowest level of the translation tree.
This gives the \bigOm{\frac{\vn p}{\vP}} part of the lower bound.
Hence, the total cost is \bigOm{\vT(d+\frac{\vn p}{\vP})}.

In the sorting phase of heapsort, we repeatedly remove the element stored in the root, move the element in the rightmost leaf to the root, and then let this element sift-down to restore the heap-property. The sift-down starts in the root and after accessing address $i$ of the heap moves to address $2i$ or $2i+1$. For the analysis, we make the additional assumption $W = M$, i.e., the data cache and the TC cache have the same size. We store the top $\ell$ layers of the heap in the data cache and the translation paths to the vertices to these layers in the TC cache, where 
$2^{\ell + 1} < M$, say $\ell = \log (M/4) = \log(W/4)$. Each of the remaining $\log n - \ell$ sift-down steps may cause $d$ cache misses. The total number of cache faults is therefore bounded by 
\[        n d (\log n - \ell) \le n \log_K(2n/P) \log (4n/W). \]
We leave the lower bound as an open problem.

\section{Cache-Oblivious Algorithms}\label{cache-oblivious algs}

Algorithms for the EM model are allowed to use the parameters of the memory hierarchy in the program code.
For any two adjacent levels of the hierarchy, there are two parameters,
the size $M$ of the faster memory and the size $B$ of the blocks in which data is transferred between the faster and the slower memory.
Cache-oblivious algorithms are formulated without reference to these parameters, i.e., they are formulated as RAM-algorithms.
Only the analysis makes use of the parameters.
A transfer of a block of memory is called an IO-operation.
For a cache-oblivious algorithm, let $C(M,B,\vn)$ be the number of IO-operations on an input of size $\vn$, where $M$ is the size of the faster memory (also called cache memory) and $B$ is the block size.
Of course, $B \leqslant M$. 

Good cache-oblivious algorithms exhibit good locality of
reference at all scales, and therefore, one may hope that they also show good
behavior in the VAT model. The following theorem gives an upper bound of
VAT-complexity in terms of the EM-complexity of an algorithm.


\begin{theorem}\label{thm: cache-to-VAT}
Consider a cache-oblivious algorithm with IO-complexity $C(M,B, \vn)$, where $M$ is the size of the cache, $B$ is the size of a block, and \vn is the input size.
Let $a\assign\lfloor\vW/d\rfloor$, where $W$ is the size of the translation cache, and let $P = 2^p$ be the size of a page.
Then, the number of cache faults in the VAT-model is at most $$\sum_{i=0}^d C(aK^i\vP,K^i\vP,\vn).$$
\end{theorem}
\begin{proof}
We divide the cache into $\vd$ parts of size $a$ and reserve one part for each level of the translation tree.

Consider any level $i$, where the leaves of the translation tree (= data pages) are on level 0.
Each node on level $i$ stands for $K^i \vP$ addresses, and we can store $a$ nodes.
Thus, the number of faults on level $i$ in the translation process is the same as the number of faults of the algorithm on blocks of size $K^i \vP$ and a memory of $a$ blocks (i.e., size $ak^i\vP$).
Therefore, the number of cache faults is at most $$\sum_{i=0}^\vd C(aK^i\vP, K^i\vP, \vn).$$
\end{proof}

Theorem~\ref{thm: cache-to-VAT} allows us to rederive some of the results from Section~\ref{analysis of algorithms}.
For example, a linear scan of an array of length $n$ has IO-complexity at most $2 + \floor{n/B}$.
Thus, the number of cache faults in the VAT-model is at most 
\[ \sum_{i=0}^d \left(2 + \frac{n}{K^i P}\right) < 2d + \frac{K}{K-1}\frac{n}{P}.\]

It also allows us to derive new results.
Quicksort has IO-complexity $O((n/B) \log (n/B))$, and hence, the number of cache faults in the VAT-model is at most 
\[ \sum_{i=0}^d \bigO{\frac{n}{K^i P} \log \frac{n}{K^i P}} = \bigO{\frac{n}{P} \log\frac{n}{P}}.\]

Binary search in the van Emde Boas layout has IO-complexity $\log_B n$, and hence, the number of cache faults in the VAT-model is at most
\begin{align*}
\sum_{i=0}^\vd \frac{\log \vn}{\log (K^i \vP)} &\le \frac{\log\vn}{\log\vP}+ \log n \int\limits_0^\vd\!\! \frac{1}{\log\vP + x\log K} dx\\
  &= \log_\vP \vn + \frac{\log n}{\log K} \ln(\log P + x \log K) \big|_0^d \\ 
  &= \log_\vP \vn + (\log_K n)(\ln(\log K^d P) - \ln (\log P))\\
  &= \log_\vP\vn+ (\log_K\vn) \ln \log_\vP(\vP K^\vd) \le  \log_\vP\vn+ \log_K\vn \ln \log_\vP 4n,
\end{align*}
where the last inequality follows from our assumption that $K^d P$ is at most twice the memory footprint of an algorithm and that the memory footprint of a binary tree with $n$ leaves is bounded by $2n$.

A matrix multiplication with a recursive layout of matrices has IO-complexity $n^3/(M^{1/2}B)$, and hence, the number of cache faults in the VAT-model is at most 
$$\sum_{i=0}^d \frac{n^3}{ (a K^i P)^{1/2} K^i P}   <   {\frac{K^{3/2}}{K^{3/2} - 1}} \frac{n^3}{a^{1/2} P^{3/2} }.$$


Cache-oblivious algorithms that match the performance of the best EM-algorithm
for the problem are known for several fundamental algorithmic problems, e.g., 
sorting, FFT, matrix multiply, and searching~\cite{FrigoEtAl12}. Do all these
algorithms automatically have small VAT-complexity via Theorem~\ref{thm:
  cache-to-VAT}? Unfortunately, the answer is \emph{no}. 
Observe that the theorem refers to the cache misses in a machine with memory
size $aK^i P$ and block size $K^i P$, i.e., memory consists of $a$
blocks. However, many of the good cache-oblivious algorithms require a
tall-cache assumption  $M \geqslant B^2$; sometimes, the assumption $M \ge B^{1 + \epsilon}$ for some positive $\epsilon$ suffices. For such algorithms, the theorem above does
not give good bounds.

\newcommand{\barM}{\overline{M}}
\newcommand{\barB}{\overline{B}}

\newcommand{\tildeM}{\widetilde{M}}
\newcommand{\tildeB}{\widetilde{B}}

\newcommand{\N}{\mathbf{N}}

In joint work with Pat Nicholson~\cite{Cache-Oblivious-VAT}, we have recently shown that cache-oblivious algorithms requiring a tall-cache assumption also perform well in the VAT-model provided a somewhat more stringent tall cache assumption holds. More precisely, consider a cache-oblivious algorithm that incurs $C(\tildeM,\tildeB,n)$ cache faults, when run on a machine with cache size $\tildeM$ and block size $\tildeB$, provided that $\tildeM \ge g(\tildeB)$. Here $g: \N \mapsto \N$ is a function that captures the ``tallness'' requirement on the cache. We consider the execution of the algorithm on a VAT-machine with cache size $\barM$ and page size $P$ and show that the number of cache faults is bounded by $4d C(M/4,dB,n)$ provided that $M \ge 4g(dB)$. Here $M = \barM/a$, $B = P/a$ and $a \ge 1$ is the size (in addressable units) of the items handled by the algorithm. 

Funnel sort~\cite{FrigoEtAl12} is an optimal cache-oblivious sorting algorithm. On an EM-machine with cache size $\tildeM$ and block size $\tildeB$, it  sorts $n$ items with 
 \[        C(\tildeM,\tildeB,n) = O\left(\frac{n}{\tildeB} \ceil{\frac{\log n/\tildeM}{\log \tildeM/\tildeB}}\right)  \]
cache faults provided that $\tildeM \ge \tildeB^2$. As a consequence of our main theorem, we obtain:

\begin{theorem} Funnel sort sorts $n$ items, each of size $a \ge 1$ addressable units, on a VAT-machine with cache size $\barM$ and page size $P$, with at most 
\[  O\left(\frac{4n}{B} \ceil{\frac{\log 4n/M}{\log M/(4dB)}}\right)\]
cache faults, where $M = \barM/a$ and $B = P/a$. This assumes $(B \log_K(2n/P))^2 \le M/4$. \end{theorem}

Since $M/(4dB) \ge (M/B)^{1/2}$ for realistic values of $M$, $B$, $K$, and $n$, this implies funnel-sort is essentially optimal also in the VAT-model.

\section{Discussion}\label{discussion}
In this section, we discuss additional topics that extend the scope of our research.
In particular, we address the comments that we received from the ALENEX13 program committee and other researchers.

\subsection{Double Address Translation on Virtual Machines}
Nowadays, more and more computation is performed on virtual machines in the clouds.
In this environment, address translation must be performed twice, first to the virtual machine addressing space and then to the host.
The cost of address translation to host can be as high as \bigO{\vT\log(\mathtt{size\ of\ virtual\ machine})}.
Moreover, big enough virtual machines may require translation for memory tables in the virtual machine, not just for the data.
This is independent of the problem input size and significant in the case of random access, but still negligible in the case of sequential access.
To test the impact of the double address translation, we timed permutation and introsort on a virtual machine; results are provided in Figure \ref{vrt-big-divnlogn}.
\begin{figure}[h]
  	\begin{center}
	\begin{tabular}{cc}
	\adjustbox{valign=m}{\begin{sideways}\mbox{running time/\vn}\end{sideways}}&
	\adjustbox{valign=m}{\includegraphics[width=0.9\textwidth]{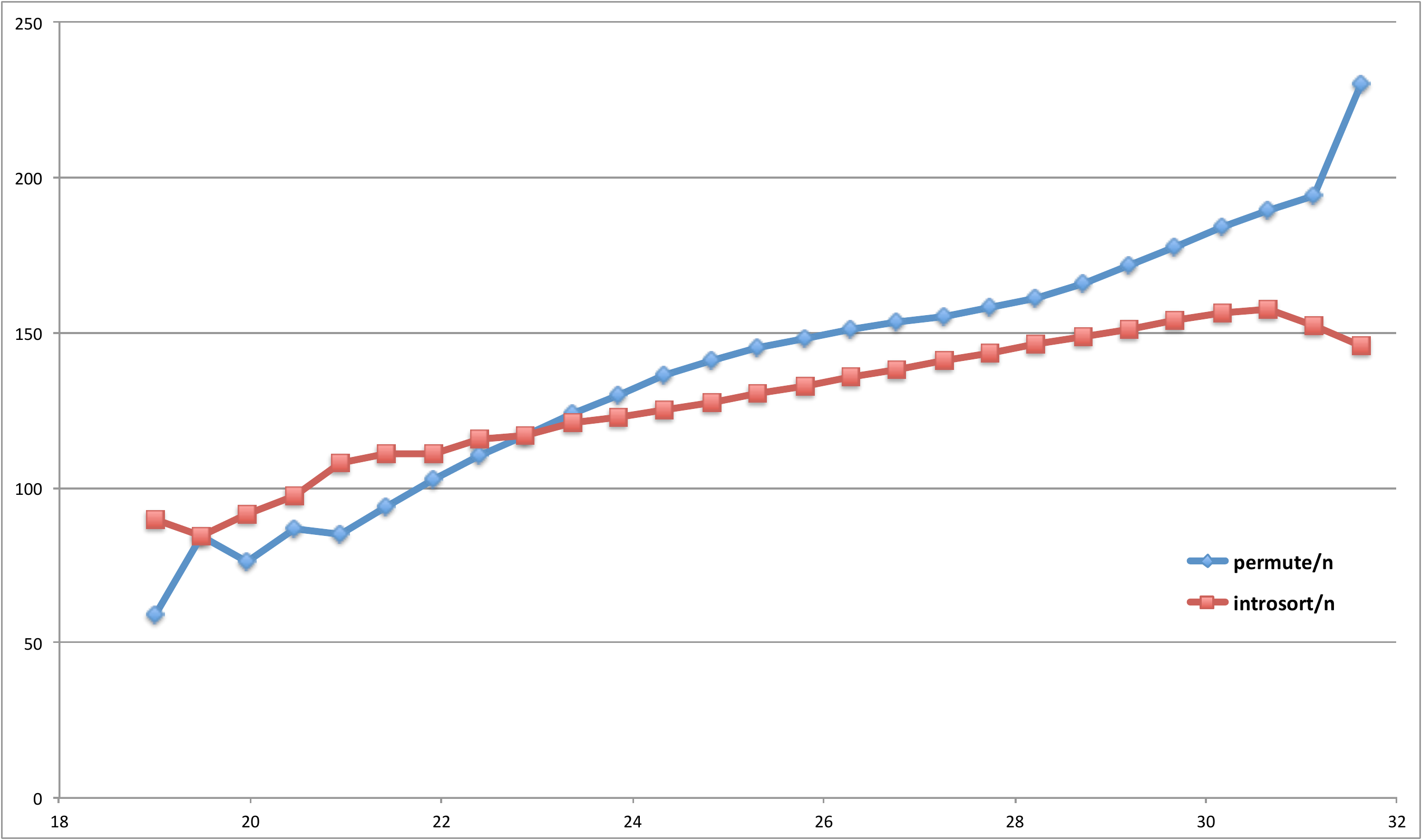}}\\
	&$\log(\text{input size})$
	\end{tabular}
  	\end{center}
\caption{Execution times divided by $\vn$ (\emph{not} the normalized operation time) on a virtual machine with the following specification:\newline
cpu: Intel(R) Xeon(R) CPU X5660 @ 2.67GHz\newline
operating system: Windows 7 Enterprise\newline
compiler: Microsoft Visual Studio 2010\newline
}\label{vrt-big-divnlogn}
\end{figure}

Please note that STL introsort takes actually less time than the permutation generator, even for very small data.
This is very surprising at first but means that a high VAT cost is especially harmful for programs launched on virtual machines.
Since many cloud systems are meant primarily for computing, the discussed phenomenon should be of primal concern for such environments.

\subsection{The Model is Too Complicated}

While we received comments that the model is too simple, we also received ones saying that the model is too complicated.
This impression is probably due to the fact that some of our proofs are somewhat technical.
Some arguments simplify if asymptotic notation is used earlier, or if the VAT cost is upper bounded by the RAM cost ahead of time (for sequential access patterns to the memory), or the other way around for the randomized access.
However, as this is the first work addressing the subject, we find it appropriate to be more detailed than absolutely necessary.
With time, more and more simplifications will appear.
Let us briefly discuss a few candidates.

\subsubsection{Value of $K$}
There is evidence that for many algorithms, the exact value of $K$ does not matter, and hence, $K = 2$ may be used.
In some cases, like repeated binary search, the exact value of $K$ seems to have only a little impact both in theory and practice.
In other cases, like permutation, it seems to be the cause of bumps on the chart in Figure \ref{data of first experiment}, but the impact is moderate.
A notable exception is matrix transpose and matrix multiplication, where the value of $K$ is blatantly visible.
The classic matrix transpose algorithm uses \bigO{\vn} operations, where \vn is the input size.
However, if the matrix is stored and read row by row, the output matrix must be written one element per row.
For a square matrix, this means a jump of $\sqrt{\vn}$ cells between writes, which means $\sqrt{\vn}$ translations of cost \bigTh{\vT\vd} to produce the first column.
As there are $\sqrt{\vn}$ translations before another element is written to the same row, no translation path can be reused if we consider the LRU algorithm.
Therefore, the total VAT cost is \bigTh{\vT\vn\vd}, which is \bigTh{\vT\vn\log\vn}.
Figure \ref{fig:matrix-transpose} shows that even though the asymptotic growth is intact, the translation cost grows in jumps rather than in a smooth logarithmic fashion.
The distance between the jumps appears to be directly related to the value of $K$, namely, the jump occurs when the matrix dimension is $K$ times greater than during the previous jump.
Note that the EM cost of this algorithm is \bigTh{\vn} for $\sqrt{\vn}\cdot\vB>\vM$, and \bigTh{\vn/\vB} for $\sqrt{\vn}\cdot\vB<\vM$.
In fact, the first cost jump is due to this barrier itself.
\begin{figure}[!h]
  	\begin{center}
	\begin{tabular}{cc}
	\adjustbox{valign=m}{\begin{sideways}\mbox{running time/RAM complexity}\end{sideways}}&
	\adjustbox{valign=m}{\includegraphics[width=0.9\textwidth]{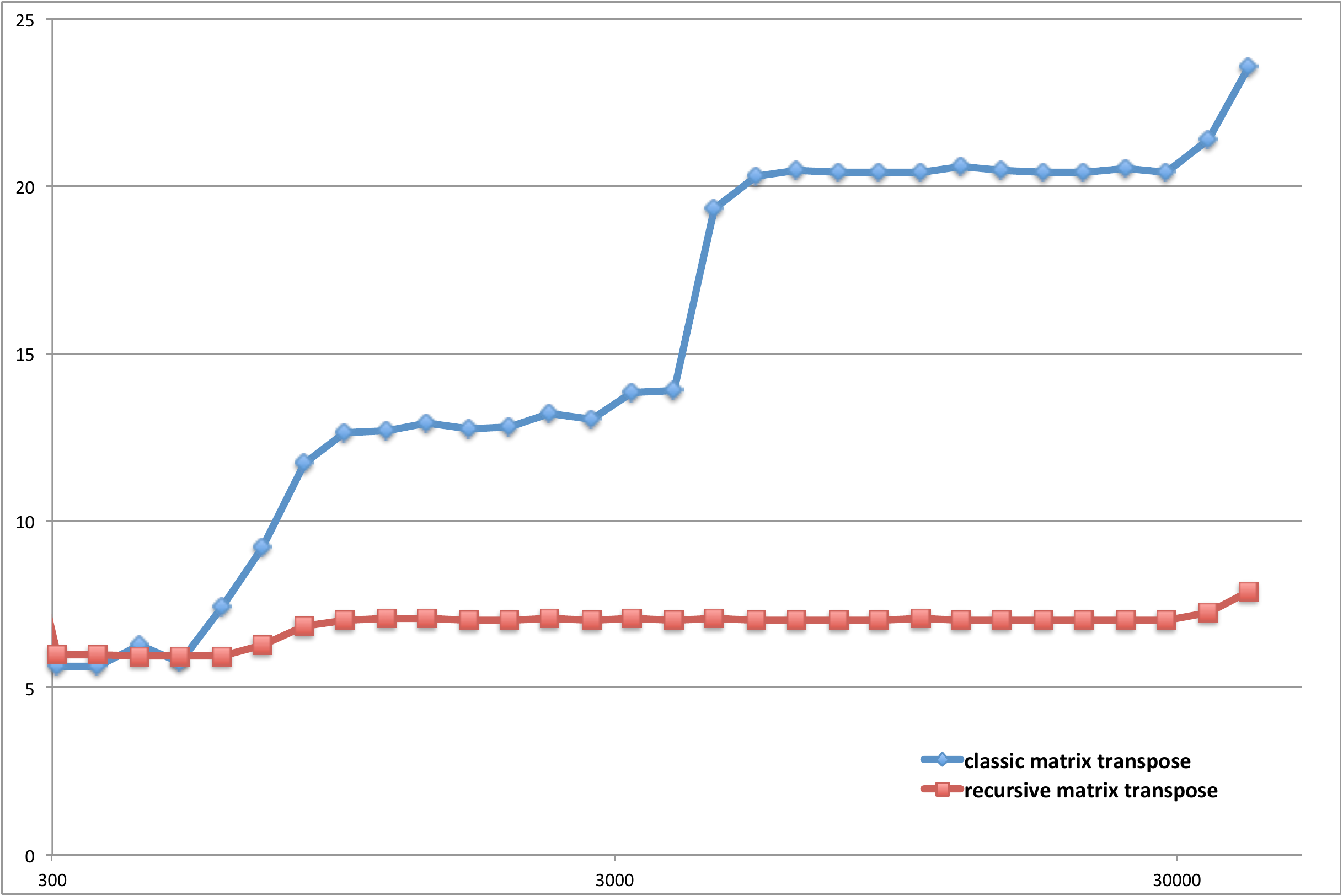}}\\
	&matrix dimension in logarithmic scale
	\end{tabular}
  	\end{center}
\caption{Running time of the matrix transpose in row by row layout and in the recursive layout}\label{fig:matrix-transpose}
\end{figure}

\subsubsection{CAT, or Sequence of Consecutive Address Translations}
In our analysis, for many algorithms precisely calculated VAT complexity was much smaller than the RAM complexity.
We believe that our approach can bring valuable insight for future research, but some of our results can be obtained in a simpler way.
The memory access patterns in the algorithms in question share some common characteristics.
There are not too few elements, they are not overspread in the memory, and the accesses are more or less performed in a sequence.
We formalize these properties in the following definition.
\begin{definition}
We call a sequence of $\ell$ memory accesses a \concept{CAT} (sequence of consecutive address translations) if it fulfills the following conditions:
\begin{itemize}
\item $\ell=\bigOm{\vT\vd}$. 
\item On average, \bigOm{\vT} elements are accessed per page in the access range.
\item The pages are used in the increasing or decreasing order. (But accesses in the page need not follow this rule).
\item Memory accesses from the sequence are separated by at most a constant number of other operations.
\end{itemize}
\end{definition}


\begin{lemma}\label{lem:cats}
In case of a CAT, the cost of the address translation is dominated by the cost of RAM operations and therefore negligible.
Hence, for CATs, it is sufficient to account for them only in the RAM part of the analysis.
\end{lemma}
\begin{proof}
We assume the LRU replacement strategy.
First, let us assess the cost of translating addresses for all the \bigO{\ell/\vT} pages in the increasing order.
The first translation causes \vd TC misses.
Since we allow only a constant number of operations between accesses from the considered sequence, the LRU replacement strategy holds translation path of the last translation when the next one starts.
Hence, the addresses to be translated change like in a classic $K$-nary counter.
The amortized cost of an update of a $K$-nary counter is \bigO{1}.
Since on average \bigOm{\vT} elements are accessed per page, the access range is at most of length \bigO{\ell/\vT}, and so the cost of updates is \bigO{\ell}.
However, we do not start counting from zero, and the potential function in the $K$-nary counter analysis can reach up to $\log{(\text{the highest number seen})}$, which in our case can reach \vd.
Hence, we need to add the cost of another \vd TC misses to our estimation.
The cost of all translations is therefore equal to $\bigO{\vT\vd+\ell}=\bigO{\ell}$.

In the definition of a CAT, we do not assume that every page is used exactly once.
However, neither skipping values in the counter, nor reusing them causes extra TC misses.

Since the RAM cost is exactly \bigTh{\ell}, it dominates the translation cost.
\end{proof}

\subsubsection{RAT, or Sequence of Random Address Translations}
Similarly, algorithms with a high VAT cost share common properties.

\begin{definition}
A sequence of memory accesses is called a \concept{RAT} (sequence of random address translations) if:
\begin{itemize}
\item There is a memory partitioning such that each part consists of all memory cells with some common virtual address prefix, and parts are of size at least $\vP\vm^\theta \text{ for }\theta\in(0,1)$.
\item For at least a constant fraction of the accesses with at least a constant probability, each access is to a part that was not accessed since \vW TC misses.
\end{itemize}
\end{definition}

\begin{lemma}\label{lem:rats}
The cost of a RAT of length $\ell$ is \bigTh{\vT \ell\vd}. It is the same as the cost of the address translations.
\end{lemma}
\begin{proof}
We assume the LRU replacement strategy.
Since parts are of size at least $\vP\vm^\theta \text{ for }\theta\in(0,1)$, a translation of an address from each part uses \bigTh{\vd} translation nodes unique to its translation subtree.
Therefore, an access to a part that was not accessed since \vW TC misses, misses the root of the subtree, and by Lemma \ref{lem:lru-many-misses}, the access causes \bigTh{\vd} misses.
As this happens for at least a constant fraction of the accesses with at least a constant probability, the total cost is \bigTh{\vT \ell\vd}.
The RAM cost is only \bigTh{\ell}, which is less than the VAT cost by the order assumption \ref{ass:tdp}.
\end{proof}

\subsection{Larger Page Sizes}
The straight-forward method to determine how the Virtual Address Translation
affects the running time of programs would be to switch it off and compare the
results. 
Unfortunately, no modern operating system provides such an option. One can
approximate the elimination of address translation by increasing the page
size. If all the data fits into a single page, address translation is
essentially eliminated. If all the data fits into a small number of pages, the
number of translations and their cost is reduced. We performed experiments with
larger page sizes. 
However, while hardware architectures support pages sized in gigabytes, operating systems do not.
Quoting \cite{hennessy}:
\begin{quote}
``\emph{Relying on the operating systems to change the page size over time.}

The Alpha architects had an elaborate plan to grow the architecture over time by growing its page size, even building it into the size of its virtual address.
When it came time to grow page sizes with later Alphas, the operating system designers balked and the virtual memory system was revised to grow the address space while maintaining the 8 KB page.

Architects of other computers noticed very high TLB miss rates, and so added multiple, larger page sizes to the TLB.
The hope was that operating systems programmers would allocate an object to the largest page that made sense, thereby preserving TLB entries.
After a decade of trying, most operating systems use these “superpages” only for handpicked functions: mapping the display memory or other I/O devices, or using very large pages for the database code.''
\end{quote}

There are good reasons why operating systems designer are reluctant to offer
larger pages. 
The main concern is space.
Pages must be correctly aligned in memory, so bigger pages lead to a greater waste of memory and limited flexibility while paging to disk.
Another problem is that since most processes are small, using bigger pages would lengthen their initialization time.
Therefore, current operating systems kernels provide only basic, nontransparent support for bigger pages.
The \emph{hugetlbpage} feature of current Linux kernels allows one  to use pages
of size 2MiB on AMD64 machines. The following links describe the
hugetlbpage-feature. 
\begin{itemize}
\item\url{http://linuxgazette.net/155/krishnakumar.html}
\item\url{https://www.kernel.org/doc/Documentation/vm/hugetlbpage.txt}
\item\url{https://www.kernel.org/doc/Documentation/vm/hugepage-shm.c}
\item\url{http://man7.org/linux/man-pages/man2/shmget.2.html}
\end{itemize}
The feature attaches a final real address one level higher in the memory table,
i.e., the last layer of nodes is eliminated from the translation trees and
pages are now of size $2^{9 + 12}$. 
This slightly decreases cache usage, decreases the number of nodes needed in each single translation but one, and finally, increases the range of addresses covered by the related entry in the TLB by $512$.

\begin{figure}[!h]
  	\begin{center}
	\begin{tabular}{cc}
	\adjustbox{valign=m}{\begin{sideways}\mbox{2MiB pages time/4kiB pages time}\end{sideways}}&
	\adjustbox{valign=m}{\includegraphics[width=0.9\textwidth]{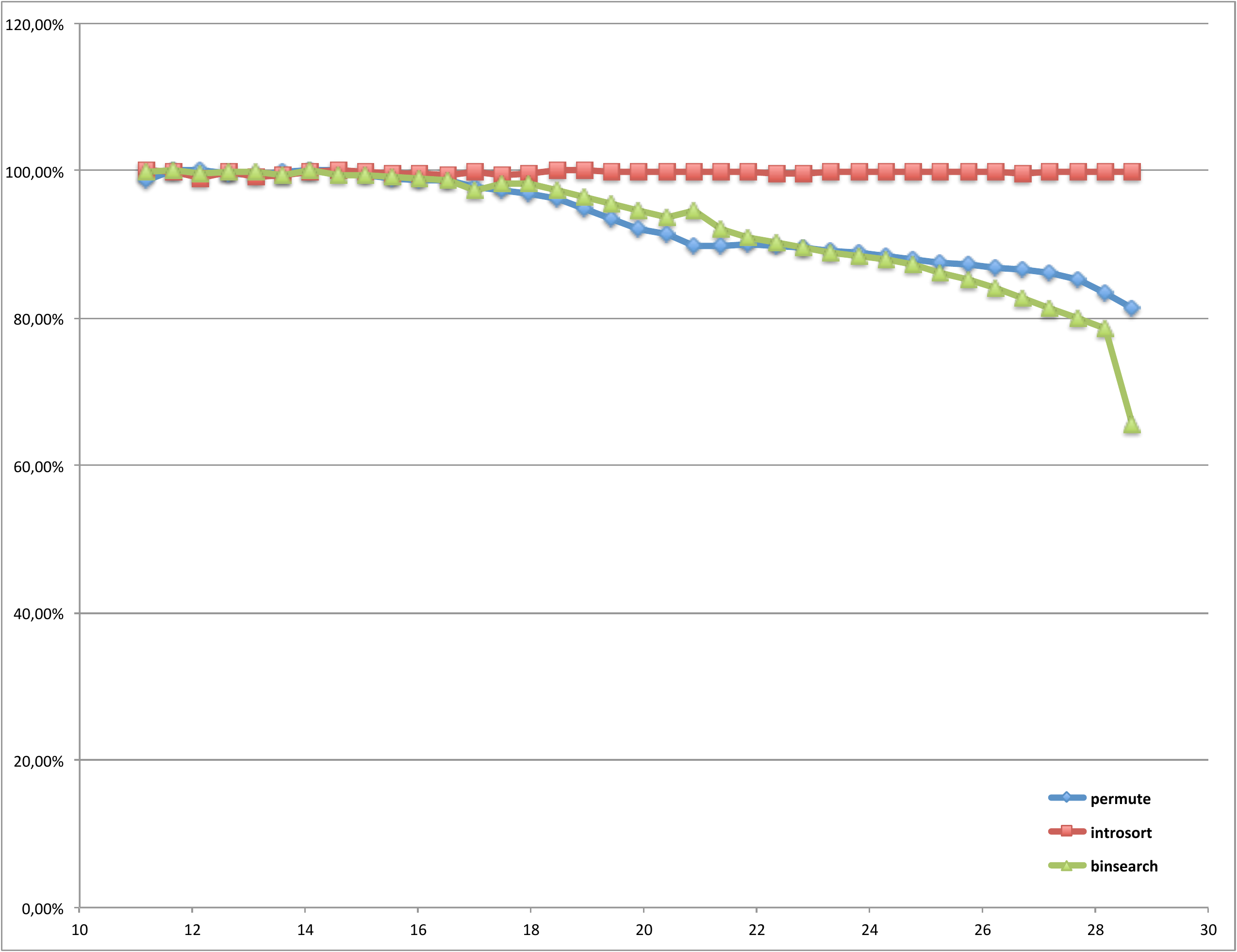}}\\
	&$\log(\text{input size})$
	\end{tabular}
  	\end{center}
\caption{
The abscissa shows the logarithm of the input size.
The ordinate shows the measured running time with use of the 2MiB pages as a percentage of the running time with 4kiB pages.}\label{fig:big-div-small}
\end{figure}

We rerun the \emph{permute}, \emph{introsort} and \emph{binsearch} on the same
machine, with and without use of the big pages. 
Figure \ref{fig:big-div-small} clearly shows that use of bigger pages can introduce a speedup.
In other words, the cost of virtual address translation can be partially
reduced by use of the bigger pages.


\subsection{The Translation Tree is Shallow}

It is true that the height of the translation tree on today's machines is bounded by $4$, and so the translation cost is bounded.
However, even though our experiments use only $3$ levels, the slowdown appears to be at least as significant in practice as the one caused by a factor of $\log\vn$ in the operational complexity.
Therefore, decreasing VAT complexity has a prominent practical significance.
Please note that while $64$ bit addresses are sufficient to address any memory that can ever be constructed according to known physics, there are other practical reasons to consider longer addresses.
Therefore, the current bound for the height of the translation tree is not absolute.


\subsection{What about Hashing?}
We have been asked whether the current virtual address translation system could be replaced with one based on hashing tables to achieve a constant amortized translation time.
Let us argue that it is not a good idea.
First and foremost, hashing tables sometimes need rehashing, and this would mean the complete blockage of an operating system.
Moreover, an adversary can try to increase a number of necessary rehashes.
Note that probabilistic guarantees are on the frequencies of the rehashes and the program isolation is insufficient to discard this concern, because an attack can be performed with side-channels like, for example, differential power analysis (see \cite{tiri}).
Finally, a tree walk is simple enough to be supported by hardware to obtain significant speedups; in case of hashing, this would be not so easy.

On the other hand, simple hash tables can be used to implement efficient caches.
In fact, associative memory can be seen as a hardware implementation of a hashing table.
If we no longer require from the associative memory that it reliable stores all the previous entries, then associative memories of small enough sizes can be well supported by hardware.
This is in fact how the TLB is implemented, and one of the reasons why it is so small.

\section{Conclusions} \label{conclusions}
We have introduced the VAT model and have analyzed some fundamental algorithms in this model.
We have shown that the predictions made by the model agree well with measured running times. 
Our work is just the beginning. In follow-up, we show together with Patrick Nicholson~\cite{Cache-Oblivious-VAT} that all cache-oblivious algorithms perform well in the VAT-model provided a tall cache assumption that is somewhat more stringent than for the EM-model. It would be interesting to know, whether this more stringent assumption is necessary. 

We believe that every data structure and algorithms course must also discuss algorithm engineering issues.
One such issue is that the RAM model ignores essential aspects of modern hardware.
The EM model and the VAT model capture additional aspects.

\bibliographystyle{ACM-Reference-Format-Journals}
\bibliography{ref,dissertation}

\received{March 2013}{??}{??}
\end{document}